%% file: Template_revised.tex
\documentclass[12pt]{article}
\usepackage{adjustbox}
\usepackage{amsfonts}
\usepackage{amsthm,amsmath,amssymb}
\usepackage[UKenglish]{babel}
\usepackage{bibunits}
\usepackage{bigints}
\usepackage{bm}
\usepackage{booktabs}
\usepackage{float}
\usepackage{graphicx}
\usepackage{geometry}
\usepackage{listings}
\usepackage{multirow}
\usepackage{siunitx}
\usepackage{tikz}
\usetikzlibrary{calc}
\usepackage{subcaption}

\newcommand{\balpha}{\bm{\alpha}}
\newcommand{\bbeta}{\bm{\beta}}

\newcommand{\balphau}{\bm{\alpha}_u}

\newcommand{\bD}{\mathcal{D}}
\newcommand{\bH}{\bm{H}}

\newcommand{\btheta}{\bm{\theta}}
\newcommand{\bX}{\bm{X}}
\newcommand{\bY}{\bm{Y}}

\newcommand{\bxo}{\bm{x}_{o}}

\newcommand{\bx}{\bm{x}}

\newcommand{\byi}{\bm{y}_{i}}
\newcommand{\bz}{\bm{z}}
\newcommand{\CRA}{\mathbb{A}_p}

\newcommand{\ebv}{\bm{\beta}_{e}}
\newcommand{\ev}{\tilde{E}^{(r)}}

\newcommand{\ismiss}{\mathbb{I}_{(k \in \mathcal{U}_i)} }
\newcommand{\isobs}{\mathbb{I}_{(k \in \mathcal{O}_i)} }

\newcommand{\mbmp}{\mathbb{B}_{p_{u}}}
\newcommand{\mbmpc}{\mathbb{B}^C_{p_{u}}}
\newcommand{\mbpi}{\mathbb{B}_{p_{u,i}}}
\newcommand{\mbmpci}{\mathbb{B}^C_{p_{u,i}}}

\newcommand{\us}{\mathcal{U}}

\newcommand{\simplex}{\mathbb{S}_p}
\newcommand{\bTD}{\mathcal{TD}}

\usepackage{xcolor}

\date{\today}

\usepackage{hyperref} %Must be added last
\graphicspath{{images/}}
\ifPDFTeX
  \usepackage[T1]{fontenc}
  \usepackage[utf8]{inputenc}
  \usepackage{textcomp} % provide euro and other symbols
\else % if luatex or xetex
  \usepackage{unicode-math}
  \defaultfontfeatures{Scale=MatchLowercase}
  \defaultfontfeatures[\rmfamily]{Ligatures=TeX,Scale=1}
\fi
\usepackage{lmodern}
\ifPDFTeX\else  
\fi
\IfFileExists{upquote.sty}{\usepackage{upquote}}{}
\makeatletter
\@ifundefined{KOMAClassName}{% if non-KOMA class
  \IfFileExists{parskip.sty}{%
    \usepackage{parskip}
  }{% else
    \setlength{\parindent}{0pt}
    \setlength{\parskip}{6pt plus 2pt minus 1pt}}
}{% if KOMA class
  \KOMAoptions{parskip=half}}
\makeatother
\usepackage{xcolor}
\makeatletter
\ifx\paragraph\undefined\else
  \let\oldparagraph\paragraph
  \renewcommand{\paragraph}{
    \@ifstar
      \xxxParagraphStar
      \xxxParagraphNoStar
  }
  \newcommand{\xxxParagraphStar}[1]{\oldparagraph*{#1}\mbox{}}
  \newcommand{\xxxParagraphNoStar}[1]{\oldparagraph{#1}\mbox{}}
\fi
\ifx\subparagraph\undefined\else
  \let\oldsubparagraph\subparagraph
  \renewcommand{\subparagraph}{
    \@ifstar
      \xxxSubParagraphStar
      \xxxSubParagraphNoStar
  }
  \newcommand{\xxxSubParagraphStar}[1]{\oldsubparagraph*{#1}\mbox{}}
  \newcommand{\xxxSubParagraphNoStar}[1]{\oldsubparagraph{#1}\mbox{}}
\fi
\makeatother

\usepackage{longtable,booktabs,array}
\usepackage{calc} % for calculating minipage widths
\usepackage{etoolbox}
\makeatletter
\patchcmd\longtable{\par}{\if@noskipsec\mbox{}\fi\par}{}{}
\makeatother
\IfFileExists{footnotehyper.sty}{\usepackage{footnotehyper}}{\usepackage{footnote}}
\makesavenoteenv{longtable}
\usepackage{graphicx}
\makeatletter
\def\maxwidth{\ifdim\Gin@nat@width>\linewidth\linewidth\else\Gin@nat@width\fi}
\def\maxheight{\ifdim\Gin@nat@height>\textheight\textheight\else\Gin@nat@height\fi}
\makeatother
\setkeys{Gin}{width=\maxwidth,height=\maxheight,keepaspectratio}
\makeatletter
\def\fps@figure{htbp}
\makeatother

\makeatletter
\@ifpackageloaded{caption}{}{\usepackage{caption}}
\AtBeginDocument{%
\ifdefined\contentsname
  \renewcommand*\contentsname{Table of contents}
\else
  \newcommand\contentsname{Table of contents}
\fi
\ifdefined\listfigurename
  \renewcommand*\listfigurename{List of Figures}
\else
  \newcommand\listfigurename{List of Figures}
\fi
\ifdefined\listtablename
  \renewcommand*\listtablename{List of Tables}
\else
  \newcommand\listtablename{List of Tables}
\fi
\ifdefined\figurename
  \renewcommand*\figurename{Figure}
\else
  \newcommand\figurename{Figure}
\fi
\ifdefined\tablename
  \renewcommand*\tablename{Table}
\else
  \newcommand\tablename{Table}
\fi
}
\@ifpackageloaded{float}{}{\usepackage{float}}
\floatstyle{ruled}
\@ifundefined{c@chapter}{\newfloat{codelisting}{h}{lop}}{\newfloat{codelisting}{h}{lop}[chapter]}
\floatname{codelisting}{Listing}

\makeatother
\makeatletter
\@ifpackageloaded{caption}{}{\usepackage{caption}}
\@ifpackageloaded{subcaption}{}{\usepackage{subcaption}}
\makeatother

\ifLuaTeX
  \usepackage{selnolig}  % disable illegal ligatures
\fi
\usepackage[]{natbib}
\usepackage{bookmark}

\IfFileExists{xurl.sty}{\usepackage{xurl}}{} % add URL line breaks if available
\hypersetup{
  pdftitle={Title},
  pdfauthor={Author 1; Author 2},
  pdfkeywords={3 to 6 keywords, that do not appear in the title},
  colorlinks=true,
  linkcolor={blue},
  filecolor={Maroon},
  citecolor={magenta},
  urlcolor={cyan},
  pdfcreator={LaTeX via pandoc}}

\newcommand{\anon}{1}

\newtheorem{theorem}{Theorem}[section]
\begin{document}

\def\spacingset#1{\renewcommand{\baselinestretch} {#1}\small\normalsize} 
%\spacingset{1}

%%%%%%%%%%%%%%%%%%%%%%%%%%%%%%%%%%%%%%%%%%%%%%%%%%%%%%%%%%%%%%%%%%%%%%%%%%%%%%
\if1\anon
{
  \title{\bf Handling Missingness and Censoring in Dirichlet Models}
  \author{%
Pillay J.\\
\shortstack[c]{%
Department of Statistics,
University of Pretoria,\\
Pretoria, South Africa
}\\[1em]
Bekker A.\\
\shortstack[c]{%
Department of Statistics,
University of Pretoria\\
National Institute for Theoretical and
Computational Sciences \\ (NITheCS),
Pretoria Node,
Pretoria, South Africa
}\\[1em]
Tortora C.\\
\shortstack[c]{%
Department of Mathematics and Statistics,
San José State University,\\
San José, USA
}\\[1em]
Punzo A.\\
\shortstack[c]{%
Department of Economics and Business,
University of Catania,\\
Catania, Italy
}%
}
  \maketitle
  \vspace{-1em}
} \fi

\if0\anon
{
  \bigskip
  \bigskip
  \bigskip
  \begin{center}
    {\LARGE\bf Handling Missingness and Censoring in Dirichlet Models}
\end{center}
  \medskip
} \fi

\bigskip
\begin{abstract}
%The text of your abstract. 200 or fewer words.

Likelihood-based inference for compositional data generally requires fully observed compositions, hindering the direct treatment of missing or censored components on the simplex. 
In this paper, we develop an expectation–maximisation (EM)-type algorithm for maximum likelihood estimation of the Dirichlet parameters in the presence of missing and censored components under a unified coarsening framework. 
The Dirichlet distribution---the canonical probability model for compositional data, which plays a role analogous to that of the multivariate normal distribution for unconstrained multivariate data---provides the foundation for our methodology. 
Our methodology preserves the compositional structure of the data while simultaneously performing parameter estimation and model-based imputation.
We evaluate the performance of our estimators and imputations through a simulation study under increasingly complex coarsening mechanisms, including both missing and censored data. 
We compare our method with an existing model-based approach and a nonparametric alternative.
Finally, we illustrate the practical utility of our methodology using mercury speciation data, in which compositions are only partially observed because of detection limits and incomplete speciation. 
Our results indicate that the Dirichlet distribution provides a suitable model for these data and that our method yields imputations that better preserve the observed compositional structure than competing approaches.

\end{abstract}

\noindent%

%\newpage
\spacingset{1.8} % DON'T change the spacing!
\maketitle

\section{Introduction}
\label{sec1}

Compositional data are multivariate observations representing the relative proportions of parts of a whole. 
A compositional vector (or, simply, composition) consists of strictly positive components that carry only relative information and satisfy a constant-sum constraint, typically equal to 1 (proportions) or 100 (percentages). 
Examples of compositional data include mineral compositions, species abundances, gene expression ratios, and budget allocations. Consequently, such data lie in the simplex.
The constant-sum constraint induces dependence among the components: an increase in one component necessarily reduces the proportion available to the others. 
Therefore, correlations observed between compositional components should not be interpreted as evidence of causal association, but rather as a mathematical consequence of the constant-sum constraint.

Modelling such data without properly accounting for this constraint induces spurious correlations and misleading results. 
Analysts typically adopt one of two approaches: either transforming the data to meet the requirements of models defined on an unconstrained space, or focusing on models that operate directly on the simplex. 
The former is achieved through transformations such as the additive log-ratio, centred log-ratio, isometric log-ratio, and square-root transformations \citep{cell_proportions, microbiome, PM2_5}. 
However, this approach has several drawbacks. 
From a methodological perspective, sufficiently complex transformations may render the interpretation of the original variables unrecoverable. 
From an interpretative standpoint, models fitted on transformed data estimate parameters that are defined in the transformed space and cannot be directly interpreted on the simplex. 
Moreover, current model-based approaches generally do not provide an injective mapping of these parameters back onto the simplex. Consequently, such models can, at best, be interpreted in the transformed space \citep{logsticnormal_inflated, logisticnormal}. 

The Dirichlet distribution is supported directly on the simplex and accounts for the dependence structure inherent in compositional data. 
Its application has been successful in numerous fields, including social media analysis \citep{dirichlet_nested_fmm}, microbiology \citep{dirichlet_microbiome, dirichlet_fmm, dirichlet_zoid, dirichlet_pathology}, and the behavioural sciences \citep{dirichlet_emotion}. 
Parameter estimation for the Dirichlet distribution is typically carried out by maximum likelihood (ML), which provides a natural and principled framework due to the exponential family structure of the model and its well-known asymptotic properties \citep{minka2000estimating, ronning1989maximum}. 
Accordingly, we adopt a likelihood-based approach in this paper.
However, compositions with unobserved components, which are common in practice \citep{geo_mean,alr_em_method,templ2011robcompositions}, cannot readily be incorporated into the ML estimation of the Dirichlet parameters. 
To accommodate the broad range of real-world situations characterised by different types and degrees of unobservedness, it is convenient to frame such components within a coarsening framework in the sense of \citet{heitjan1991ignorability}. 
Such coarsening may arise from various sources, including technical failures in recording devices, respondent non-compliance, incomplete data records, detection limits of measurement instruments, or imprecise self-reporting in surveys. 
Under this framework, instead of observing the exact value of a component, one observes that it belongs to a feasible set. 

A particularly important form of coarsening is censoring, whereby the true value is known to lie within a specific interval but is not observed exactly.
Censored components can be classified as left-censored (value below a threshold), right-censored (value above a threshold), or interval-censored (value lying within a known interval); in all cases, partial information about the underlying value is retained. 
Missing components can be viewed as an extreme form of censoring and are therefore included in our framework, with the feasible set coinciding with the entire support of the variable. 
In this case, no direct information about the component is available beyond that implied by the compositional constraint and the observed components.

The coarsening mechanism plays a crucial role in statistical inference. 
Within the general coarsening framework of \citet{heitjan1991ignorability}, one distinguishes between coarsening completely at random (CCAR), coarsening at random (CAR), and coarsening not at random (CNAR).  
Under CCAR, the coarsening mechanism is independent of the underlying true values. Under CAR, the mechanism may depend on the observed coarsened data but not on the unobserved exact values within the coarsened sets; in this case, the coarsening mechanism is typically ignorable for likelihood-based inference. Finally, under CNAR, the mechanism depends on the unobserved values themselves and must therefore be modeled explicitly.
This framework encompasses classical missing-data formulations as a special case. 
In particular, the standard distinctions---see, for example, \citet{missingness}---between missing completely at random (MCAR), missing at random (MAR), and missing not at random (MNAR) correspond, respectively, to CCAR, CAR, and CNAR when missingness is viewed as a form of coarsening \citep{little2021missing}. 

%Among the various forms of coarsening, we focus on censoring and missing data, which we collectively refer to as unobserved components.

%Within this framework, inference is based on the observed coarsened likelihood, and the coarsening mechanism can be treated as ignorable under a coarsening-at-random (CAR) assumption \citep{heitjan1991ignorability}. MAR is a special case of CAR \citep{little2021missing}.

The occurrence of unobserved components is sufficiently frequent to warrant the development of various techniques for handling incomplete datasets. Typical pre-processing approaches are broadly classified into three categories: deletion, imputation, and substitution. Deletion, which is valid under a CCAR mechanism, discards the entire row (composition) of the data matrix if it is not fully observed. 
It excludes potentially meaningful information, thereby weakening the representativeness of the data \citep{food_missing}. Imputation, which typically relies on CCAR or CAR assumptions, is the process of replacing unobserved components with plausible values, and can be further categorised into parametric and non-parametric approaches. One of the longstanding methods of the latter type is based on the geometric mean and is known as multiplicative replacement \citep{geo_mean}. 
Other univariate methods include $k$-nearest neighbours ($k$NN), least squares, and trimmed least squares imputation \citep{knn, svd}. The popularity of non-parametric methods stems from their convenience, as they require minimal assumptions for implementation. However, these approaches are typically designed for missing values and do not explicitly account for censoring mechanisms.

A particularly popular model-based imputation strategy involves transforming the data to a space in which algorithms for fitting probability distributions that accommodate incomplete observations are available \citep{alr_em_method}. 
Imputations obtained in the transformed space must then be mapped back to the simplex; this inversion is often challenging and may amplify errors introduced during transformation. Moreover, censoring is not naturally accommodated within such transformation-based approaches. Finally, substitution, that is, the practice of replacing unobserved entries with fixed values—most commonly zeros—is a widely used approach in compositional data analysis, particularly in the presence of detection limits.

All three approaches may distort the underlying distribution of the data if their implicit assumptions are violated, leading to biased parameter estimates and reduced interpretability (see \citealp{review} for an overview of the literature). To address these limitations, this paper proposes an algorithm that performs simultaneous ML estimation of the Dirichlet parameters and imputation of CAR unobserved values.

Handling unobserved values on the simplex, as proposed, has important advantages:
\begin{enumerate}

\item The algorithm enables the fitting of a Dirichlet distribution when the dataset is not fully observed, requiring only a single observed component per row.

\item The approach accommodates both missing and censored observations—collectively referred to as unobserved components—within a unified modelling framework.

\item An Expectation--Maximisation (EM) type algorithm is developed to obtain ML estimates of the Dirichlet parameters in the presence of unobserved CAR entries. 
The complete-data and observed-data likelihoods are derived, and an efficient update procedure is obtained, yielding a computationally attractive framework well suited to high-dimensional settings.

\item The method performs parameter estimation and imputation jointly within a single likelihood-based framework, operating directly on the simplex without requiring transformations to the real space, thereby avoiding the additional errors and interpretational limitations associated with back-transformation.

\item An extensive simulation study evaluates the proposed method in comparison with existing approaches, assessing parameter recovery and  imputation accuracy across varying levels of missingness.

\item An application to a real dataset demonstrate improved interpretability and consistent inference on the original compositional variables.

\end{enumerate}

The rest of the paper is structured as follows. Section \ref{prelim} introduces the Dirichlet distribution and its theoretical properties, as well as the types of unobserved entries considered in this paper. 
The main contribution is in Section \ref{inference}, i.e., the algorithm for fitting the Dirichlet distribution to incomplete data. 
A practical application of the proposed method is shown in Section \ref{application}, followed by concluding remarks in Section \ref{conclusion}.  An extensive simulation study is shown in the Supplementary Material \ref{simulations}.

\section{Preliminaries}
\label{prelim}

This section introduces the theoretical framework underlying the proposed methodology. We first review the Dirichlet and truncated Dirichlet distributions (Sections~\ref{dirichlet_review} and~\ref{truncated_dirichlet_review}, respectively). 
We then formalise the treatment of missing and censored components within a unified coarsening framework (Section~\ref{subsec:Presence_of_unobserved_components}) and derive the key distributional results that form the basis of the proposed likelihood-based estimation procedure (Section~\ref{subsec:conditional_distribution_unobserved}).

Throughout this paper, for a given dimension $p \in \mathbb{N}_{+}$, we introduce and frequently use the following domains for compositional data analysis: the $p$-dimensional open unit simplex
   $ \mathbb{V}_p 
    = \left\{ \bx = [x_1,\ldots,x_p]^{\top} 
    : \text{$x_k \in (0,1)$, $k \in \{1,\ldots,p\}$, $\sum_{k=1}^p x_k < 1$} \right\},$ and the $p$-dimensional closed unit simplex
\begin{equation*}
    \mathbb{S}_p 
    = \left\{ \bx = [x_1,\dots,x_p]^{\top} 
    : \text{$x_k \in (0,1)$, $k \in \{1,\ldots,p\}$, $\sum_{k=1}^p x_k = 1$}
     \right\}.
\end{equation*}

\subsection{The Dirichlet distribution}
\label{dirichlet_review}

A random vector $\bX\in\mathbb{S}_p$ is said to follow a Dirichlet distribution with parameter vector $\balpha = [\alpha_1,\ldots,\alpha_p]^{\top} \in \mathbb{R}_+^p$, denoted by $\bX \sim \bD_{\mathbb{S}_p}(\balpha)$, if its probability density function (density) is
\begin{equation}
\label{pdf_dirichlet}
    f_{\bD_{\mathbb{S}_p}}(\bx;\balpha) 
    =
       \dfrac{\Gamma(\alpha_0)}
        {\displaystyle\prod_{k=1}^p \Gamma(\alpha_k)}
        \displaystyle\prod_{k=1}^p x_k^{\alpha_k-1},\qquad 
        \bx \in \mathbb{S}_p, 
\end{equation}
where $\alpha_0 = \|\balpha\|_1 = \displaystyle\sum_{k=1}^p \alpha_k$, and $\Gamma(\cdot)$ denotes the gamma function.
Since $\displaystyle\sum_{k=1}^p X_k = 1$, the last component of $\bX$---though, in principle, any component could be chosen---can be expressed as
$
X_p = 1 - \sum_{k=1}^{p-1} X_k.
$
Hence, $\bX \sim \bD_{\mathbb{S}_p}(\balpha)$ can equivalently be represented by its first $p-1$ components, say $\bX_{-p} = [X_1,\ldots,X_{p-1}]^{\top}\in\mathbb{V}_{p-1}$, whose density is
\begin{equation}
    f_{\bD_{\mathbb{V}_{p-1}}}(\bx_{-p};\balpha) 
    =
        \dfrac{\Gamma(\alpha_0)}
        {\displaystyle\prod_{k=1}^p \Gamma(\alpha_k)}
        \left(1 - \|\bx_{-p}\|_1 \right)^{\alpha_p-1}\displaystyle\prod_{k=1}^{p-1} x_k^{\alpha_k-1}
        , \qquad
        \bx_{-p} \in \mathbb{V}_{p-1}. 
        \label{eqref:Dirichlet2}
\end{equation}
In this case, we will write $\bX_{-p} \sim \bD_{\mathbb{V}_{p-1}}(\balpha)$.
Note that, although \eqref{pdf_dirichlet} and \eqref{eqref:Dirichlet2} represent
equivalent formulations of the same density, they are defined on different domains, namely $\mathbb{S}_p$ and $\mathbb{V}_{p-1}$,
respectively. 
See \citet[][Chapter~2]{ng2011dirichlet} for further details on the Dirichlet distribution.

\subsection{The truncated Dirichlet distribution}
\label{truncated_dirichlet_review}

Let $\CRA \subset \mathbb{S}_p$ be a measurable set with
$\Pr\left(\bY \in \CRA\right) > 0$, where $\bY \sim \bD_{\mathbb{S}_p} (\balpha)$.
A random vector $\bX$ is said to follow a truncated Dirichlet distribution on $\CRA$, with parameter vector $\balpha \in \mathbb{R}_+^p$,
denoted $\bX \sim \bTD_{\CRA}(\balpha)$, if $\bX \stackrel{d}{=} \bY \mid \{\bY \in \CRA\}.$ The density of $\bX \sim \bTD_{\CRA}(\balpha)$ is
\begin{align}
\label{truncated_pdf}
f_{\bTD_{\CRA}}(\bx;\balpha) 
=
\frac{f_{\bD_{\mathbb{S}_p}}(\bx;\balpha)}{F(\CRA; \balpha)},
%\, \mathbb{I}_{\CRA}(\bx),
%\quad \bx \in \mathbb{S}_p,
\qquad \bx \in \CRA,
\end{align}
where
$
F(\CRA; \balpha)
=
\Pr(\bY \in \CRA), \bY \sim \bD_{\mathbb{S}_p}(\balpha),
$
is the Dirichlet probability of the truncation region $\CRA$.
Throughout this paper, $\CRA$ is assumed to be defined by lower and/or upper bounds on the components of the random vector. Specifically,
\begin{equation}
    \CRA = \left\{ \bx \in \mathbb{S}_p 
: a_{kL} \le x_k \le a_{kU}, \; k = 1,\dots,p \right\} \nonumber = \left\{ \bx \in \mathbb{S}_p 
: \bm{a}_L \le \bx \le \bm{a}_U \right\}, \label{eq:Ap}
\end{equation}

where $\bm{a}_L = [a_{1L},\ldots,a_{pL}]^\top$ and 
$\bm{a}_U = [a_{1U},\ldots,a_{pU}]^\top$ denote the vectors of lower and upper bounds, respectively. 
The bounds satisfy $0 \le a_{kL} < a_{kU} \le 1$ for $k=1,\ldots,p$, and the inequalities are understood componentwise.
See \citet[][Chapter~7]{ng2011dirichlet} for further details on the truncated Dirichlet distribution.

\subsection{Presence of unobserved components}
\label{subsec:Presence_of_unobserved_components}

As discussed in Section \ref{sec1}, in compositional data analysis it is common to encounter compositions (i.e., rows of the data matrix) with unobserved components (i.e., entries of the data matrix). 
As in any incomplete-data setting, a composition is informative for inference only if at least one component is observed. 
However, unlike in general multivariate settings, the compositional constraint $\displaystyle\sum_{k=1}^p X_k = 1$ implies that whenever exactly one component is unobserved, its value is uniquely determined by the observed ones. 
Therefore, non-identifiability—and hence the need for statistical inference—arises only when at least two components are unobserved, a distinctive feature of the compositional framework. 
Consequently, meaningful inference in this setting requires compositions with at least four components, that is, $p \geq 3$.

Always as discussed in Section \ref{sec1}, unobserved components may be either completely unobserved (i.e., missing) or partially unobserved (i.e., censored, in the framework of this paper). 
Properly accounting for these forms of incompleteness is therefore essential when modeling Dirichlet-distributed data.
To this end, we partition the random vector $\bX$ into its unobserved and observed components,
$
\bX =
\begin{bmatrix}
\bX_u \\
\bX_o
\end{bmatrix},
$
where the superscripts $u$ and $o$ denote the unobserved and observed parts, respectively.
Let $\mathcal{U} \subseteq \{1,\dots,p\}$ denote the index set of unobserved components, with cardinality $p_u = |\mathcal{U}|$, and let $\mathcal{O} = \{1,\dots,p\} \setminus \mathcal{U}$ denote the index set of observed components, the complement set of $\mathcal{U}$, with cardinality $p_o = |\mathcal{O}|$. 
Accordingly, $\bX_u$ and $\bX_o$ denote the subvectors of $\bX$ formed by the components indexed by $\mathcal{U}$ and $\mathcal{O}$, respectively.

To operationally handle compositions with unobserved components, and to define the region $\bX_u$ belongs to, we introduce the set
\begin{equation}
\mathbb{B}_{p_u}
=
\left\{
\bx_u \in \mathbb{V}_{p_u}
:
b_{kL} \le x_k \le b_{kU} ,
\; k \in \mathcal{U}
\right\}, \label{eq:Bpu}
\end{equation}
where, in line with \eqref{eq:Ap}, $0 \le b_{kL}< b_{kU}\le 1$ denote the lower and upper bounds for the $k^{\text{th}}$ unobserved component.
According to this notation, a generic unobserved component $x_k$, $k \in \mathcal{U}$, can be represented in terms of the interval $[b_{kL}, b_{kU}]$: a missing component corresponds to $[0,1]$, a left-censored component to $[0, b_{kU}]$, a right-censored component to $[b_{kL}, 1]$, and an interval-censored component to $[b_{kL}, b_{kU}]$.
Thus, missing components can be viewed as a special case of censored components with $b_{kL}=0$ and $b_{kU}=1$.
In terms of available information, missing and censored components represent two distinct forms of incompleteness: in the former case, no information is available on the component, whereas in the latter, partial information is retained through known bounds.

To better understand the information provided by each type of unobserved component, \figurename~\ref{fig:example} shows, on the ternary diagram, four compositions in $\mathbb{S}_3$ exhibiting increasing levels of incompleteness and, correspondingly, decreasing levels of information.
\input Triangles.tex
The four compositions are
\[
\bx_1=
\begin{bmatrix}
0.2 \\ 0.4 \\ 0.4
\end{bmatrix},
\quad
\bx_2=
\begin{bmatrix}
0.2 \\ [0.2,0.4] \\ [0.3,0.5]
\end{bmatrix},
\quad
\bx_3=
\begin{bmatrix}
0.2 \\ [0.2,0.4] \\ [0,1]
\end{bmatrix},
\quad \text{and}\quad
\bx_4=
\begin{bmatrix}
0.2 \\ [0,1] \\ [0,1]
\end{bmatrix}.
\]
The first point $\bx_1$, shown in \figurename~\ref{fig:ex1}, is a fully observed composition, whose location in the simplex is uniquely identified; this represents the maximum amount of information that a composition can convey.
The second point $\bx_2$, shown in \figurename~\ref{fig:ex2}, is partially observed: the first component is known exactly, whereas the remaining two are interval-censored. In this case, the exact position of the composition cannot be determined; rather, the set of feasible values is restricted to the blue segment in the simplex, reflecting the uncertainty induced by censoring.
The third point $\bx_3$, shown in \figurename~\ref{fig:ex3}, corresponds to a composition in which one component is interval-censored and another is completely unobserved (i.e., missing). 
Compared to the previous case, the available information is further reduced, and the feasible set expands accordingly, covering a larger portion of the simplex compatible with both the censoring bounds and the compositional constraint.
Finally, the fourth point $\bx_4$, shown in \figurename~\ref{fig:ex4}, exhibits two completely unobserved (missing) components and only one observed value. 
In this situation, the uncertainty is maximal among the incomplete cases considered: the feasible region spans the largest admissible subset of the simplex consistent with the observed component and the unit-sum constraint.

These examples illustrate how different types and degrees of incompleteness translate into feasible regions of varying size on the simplex, and highlight the progressive loss of information as the number of unobserved components increases and interval-censored components are replaced by completely unobserved ones.
Consequently, the model in \eqref{pdf_dirichlet} must be adapted to account for the fact that, rather than observing a single point in $\mathbb{S}_p$, one observes a region of admissible values.

\subsection{Conditional distribution of unobserved components}
\label{subsec:conditional_distribution_unobserved}

Under the above formulation, the conditional distribution of the unobserved components $\bX_u$ given the observed ones $\bX_o$ admits a closed-form expression. Further, the marginal distribution of the observed parts is now also conditioned the information that the unobserved parts exist within a region, namely $\mbmp$. Theorems ~\ref{marginal_dist} and \ref{cond_dist} formalise these results and constitute key contributions, as they enable tractable likelihood-based inference in the presence of incomplete compositions (see Section~\ref{inference}).
%%%%%%%%%%%%%%%%%%
\begin{theorem}
\label{marginal_dist}
Let $\bX \sim \bD_{\mathbb{S}_p}(\balpha)$, and partition $\bX$ and $\balpha$ according to the index sets $\mathcal{U}$ and $\mathcal{O}$ as

$\bX =
\begin{bmatrix}
\bX_u \\
\bX_o
\end{bmatrix}
\quad\text{and}\quad
\balpha =
\begin{bmatrix}
\balpha_u \\
\balpha_o
\end{bmatrix}.$ Let $p_u = |\mathcal{U}|$ and $p_o = |\mathcal{O}|$. Define $\balpha_o^*
=
\begin{bmatrix}
\balpha_o \\
\|\balpha\|_1 - \|\balpha_o\|_1
\end{bmatrix}.$ Then the marginal density of the observed component $\bX_o$ is given by
\begin{equation}
\label{marginal}
f(\bxo;\balpha)
=
f_{ {\bD_{\mathbb{V}_{p_o-1}}} }(\bxo;\balpha_o^*)\;
F_{\bD_{\mathbb{S}_{p_u}}}\left(\mbmp;\balpha_u \mid \bxo\right),
\qquad \bxo \in \mathbb{V}_{p_o}, %\\
\end{equation}
where $\mbmp$ is defined as in \eqref{eq:Bpu} and
$F_{\bD_{\mathbb{S}_{p_u}}}\left(\mbmp;\balpha_u \mid \bxo\right)
=
\Pr\left(\bX_u \in \mbmp \,\middle|\, \bX_o=\bxo\right).$

% &=
% \Pr\left(\frac{\bX_u}{c(\bxo)} \in \mbmpc \,\middle|\, \blue{\bX_o=\bxo}\right).
%\end{equation*}
\end{theorem}
%%%%%%%%%%%%%
\begin{proof}
To show the result in \eqref{marginal}, let us consider
\begin{align}
f(\bm{x}_o;\balpha) 
& = \int_{\mathbb{B}_{p_u}} f_{\mathcal{D}_{\mathbb{S}_p}}(\bm{x}; \bm{\alpha}) \, d\bm{x}_u 
\nonumber\\
%& = \int_{\mathbb{B}_{p_u}} f_{\mathcal{D}_{\mathbb{S}_p}}(\bm{x}_o, \bm{x}_u; \bm{\alpha}) \, d\bm{x}_u \nonumber\\
%& = \int_{\mathbb{B}_{p_u}} f(\bm{x}_u \mid \bm{x}_o)\, f(\bm{x}_o)\, d\bm{x}_u \nonumber\\
&= f_{_{\bD_{\mathbb{V}_{p_o-1}}}}(\bm{x}_o; \bm{\alpha}_o^*) 
\int_{\mathbb{B}_{p_u}} f(\bm{x}_u \mid \bm{x}_o)\, d\bm{x}_u \nonumber\\ 
&= f_{{\bD_{\mathbb{V}_{p_o-1}}}}(\bm{x}_o; \bm{\alpha}_o^*) \Pr\left(\bm{X}_u \in \mathbb{B}_{p_u} \mid \bm{x}_o\right). 
\label{density}
%&= f_\mathcal{D}(\bm{x}_o; \bm{\alpha}^*) 
%\; F(\bm{x}_u \in \mathbb{B}_{p_u} \mid \bm{x}_o)
\end{align}

Denote the closure factor associated with $\bxo$ as $c(\bxo)=1-\|\bxo\|_1$.
Then
\begin{align}
\Pr\left(\bm{X}_u \in \mathbb{B}_{p_u} \mid \bm{x}_o\right) 
%= \Pr\left(\bm{b}_L \le \bm{X}_u \le \bm{b}_U \mid \bm{x}_o\right)
= \Pr\left(
\frac{\bm{b}_L}{c(\bxo)}
\le 
\frac{\bm{X}_u}{c(\bxo)}
\le 
\frac{\bm{b}_U}{c(\bxo)}
\;\middle|\; \bm{x}_o
\right).
\label{prob_x}
\end{align}
To make the notation more concise, let $\bm{Y} = \bm{X}_u/c(\bxo)$.
% \begin{equation*}
%     \bm{Y} = \frac{\bm{X}_u}{1 - \|\bm{x}_o \|_1}.
% \end{equation*}
Further, define the following intervals of $\bm{Y}$ for all $k \in \mathcal{U}$ as
$ y_{kL} = \max\left\{0, \frac{b_{kL}}{c(\bxo)} \right\} \quad \text{and}\quad y_{kU} = \min\left\{1, \frac{b_{kU}}{c(\bxo)} \right\}.   $

By \citet[Theorem~2.5]{ng2011dirichlet}, $\bY \sim \bD_{\mathbb{S}_{p_u}} (\balphau)$. 
Therefore, \eqref{prob_x} can be rewritten in terms of $\bY$ as 
$\Pr\left(\bm{X}_u \in \mathbb{B}_{p_u} \mid \bm{x}_o\right) = \Pr\left(\bm{Y} \in \mbmpc \mid \bm{x}_o\right) \nonumber = F_{\bD_{\mathbb{S}_{p_u}}}(\mathbb{B}^C_{p_u}; \bm{\alpha}_u, \bm{x}_o).$ By substituting this expression into the probability into \eqref{density} we get
\begin{equation*}
f(\bm{x}_o; \bm{\alpha})
=
f_{\bD_{\mathbb{V}_{p_o-1}}}(\bm{x}_o; \bm{\alpha}_o^*) 
\, F_{\bD_{\mathbb{S}_{p_u}}}(\mathbb{B}^C_{p_u}; \bm{\alpha}_u, \bm{x}_o),
\qquad
\bm{x}_u \in \mathbb{B}_{p_u}. 
\end{equation*}
\end{proof}

\begin{theorem}
\label{cond_dist}
Let $\bX \sim \bD_{\mathbb{S}_p}(\balpha)$, and partition $\bX$ and $\balpha$ according to the index sets $\mathcal{U}$ and $\mathcal{O}$ as
$
\bX =
\begin{bmatrix}
\bX_u \\
\bX_o
\end{bmatrix}
\text{ and }
\balpha =
\begin{bmatrix}
\balpha_u \\
\balpha_o
\end{bmatrix}.
$
Let $p_u = |\mathcal{U}|$ and $p_o = |\mathcal{O}|$. Moreover, letting $c(\bxo)=1-\|\bxo\|_1>0$, and $\bX_u$ having a censored region $\mbmp$, we have:
\begin{equation}
\label{cond}
\frac{\bX_u}{c(\bxo)}
\;\bigg|\;
\bX_o=\bxo
\sim
\bTD_{\mbmpc}(\balpha_u),
\end{equation}
where 
$
\mbmpc
=
\left\{
\bz \in \mathbb{S}_{p_u}
:
\frac{b_{kL}}{c(\bxo)} \le z_k \le
\min\left\{1,\frac{b_{kU}}{c(\bxo)}\right\},
\; k \in \mathcal{U}
\right\}$ denotes the rescaled version of $\mbmp$ induced by the closure factor $c(\bxo)$.
\end{theorem}
%%%%%%%%%%%%%
\begin{proof}

First, recognise that $\frac{\bm{X}_u}{c(\bxo)} \sim \bD_{\mathbb{S}_{p_u}} (\balphau)$. The proof of this result can be found in Theorem 2.5 of \citet[Section 2.2]{ng2011dirichlet}. All that is left is to recognise that the random variable's support is truncated, and thus, its distribution's support must also be truncated. Following Definition \ref{truncated_pdf}, it follows that $\frac{\bm{X}_u}{c(\bxo)} \sim \bTD_{\mbmpc}(\balphau)$, as required.
\end{proof}

Theorem \ref{cond_dist} is necessary for deriving closed-form expressions of expected values required for the E-step in new algorithm developed in Section \ref{inference}. Theorem \ref{marginal_dist} reveals that the marginal distribution of the observed parts are scaled versions of the marginal Dirichlet distributions. This fact is useful for assessing and tracking the convergence of the proposed algorithm in Section \ref{inference}.

\section{EM for Maximum Likelihood Estimation}
\label{inference}

Maximum likelihood (ML) estimation of the Dirichlet parameter vector $\balpha$ from a random sample
$\mathcal{X}=\{\bx_i: i=1,\dots,n\}$, with $\bx_i=[x_{i1},\dots,x_{ip}]^\top\in\mathbb{S}_p$, becomes nontrivial when some components of the compositions are censored or missing. 
In this setting, ML parameter estimation can be carried out using an expectation-maximisation (EM) type algorithm, which is well suited to likelihood maximisation in the presence of incomplete data \citep{em_book}.

\subsection{EM algorithm}
\label{ecm}

The EM algorithm treats the unobserved components as latent variables and
iteratively maximises the conditional expectation of the complete-data
log-likelihood. As introduced in
Section~\ref{subsec:Presence_of_unobserved_components}, each observation
$\bx_i$ is partitioned into unobserved and observed subvectors,
$
\bx_i =
\begin{bmatrix}
\bx_{u,i}\\
\bx_{o,i}
\end{bmatrix},
$
where the subscripts $u$ and $o$ refer to the unobserved and observed components, respectively. 
Let $\mathcal{U}_i \subseteq \{1,\dots,p\}$ denote the index set of unobserved components of $\bx_i$, with cardinality $p_{u,i}$, and let $\mathcal{O}_i = \{1,\dots,p\} \setminus \mathcal{U}_i$ denote the index set of observed components, the complement set of $\mathcal{U}_i$, with cardinality $p_{o,i}$. 
The unobserved subvector lies in a feasible region, denoted by $\mbpi$.

Under the assumed CAR mechanism (refer to Section~\ref{sec1}), the coarsening mechanism is ignorable for likelihood-based inference on $\balpha$. 
Hence, inference can be based on the
observed-data likelihood induced by the complete-data Dirichlet model. 
Based on \ref{pdf_dirichlet}, the corresponding complete-data likelihood is
\begin{align}
\label{likelihood}
\mathcal{L}_c(\balpha;\mathcal{X})
&
%=\prod_{i=1}^n f_{\bD_{\mathbb{S}_p}}(\bx_i;\balpha) 
%&= \prod_{i=1}^n
%\left\{
%\frac{\Gamma(\alpha_0)}
%{\displaystyle\prod_{k=1}^p\Gamma(\alpha_k)}
%\prod_{k=1}^p x_{ik}^{\alpha_k-1}
%\right\} \nonumber\\
=
\prod_{i=1}^n
\left\{
\frac{\Gamma(\alpha_0)}
{\displaystyle\prod_{k=1}^p\Gamma(\alpha_k)}
\left(\prod_{k\in\mathcal{U}_i}x_{ik}^{\alpha_k-1}\right)
\left(\prod_{k\in\mathcal{O}_i}x_{ik}^{\alpha_k-1}\right)
\right\}.
\end{align}
 % $\mathcal{U}_i\subseteq\{1,\dots,p\}$ and $\mathcal{O}_i=\{1,\dots,p\}\setminus\mathcal{U}_i$ denote the index sets of unobserved and observed components in $\bx_i$, respectively.
Taking logarithms yields the complete-data log-likelihood
\begin{align}
\label{ll}
\ell_c(\balpha;\mathcal{X})
%&=
%\sum_{i=1}^n
%\Bigg\{
%\ln\Gamma(\alpha_0)
%-
%\sum_{k=1}^p\ln\Gamma(\alpha_k)
%+
%\sum_{k\in\mathcal{U}_i}(\alpha_k-1)\ln x_{ik}
%+
%\sum_{k\in\mathcal{O}_i}(\alpha_k-1)\ln x_{ik}
%\Bigg\}
%\nonumber\\
&=
n\ln\Gamma(\alpha_0)
-
n\sum_{k=1}^p\ln\Gamma(\alpha_k)
+
\sum_{i=1}^n
\left(
\sum_{k\in\mathcal{U}_i}(\alpha_k-1)\ln x_{ik}
+
\sum_{k\in\mathcal{O}_i}(\alpha_k-1)\ln x_{ik}
\right).
\end{align}
The complete-data log-likelihood in \eqref{ll} depends on the unknown parameter vector $\balpha$ and also involves the unobserved components $\bx_{u,i}$, which must be imputed for $i=1,\dots,n$. 
The EM algorithm alternates between an E-step and an M-step until convergence to stable parameter estimates.
The steps corresponding to the $(r+1)^{\text{th}}$ iteration of the algorithm are described below.

\subsubsection{E-step}
\label{e_step}

The E-step computes the conditional expectation of the complete-data log-likelihood,
$
Q(\balpha)=
\mathbb{E}\!\left[
\ell_c(\balpha;\mathcal{X})
\mid
\balpha^{(r)},\mathcal{X}_o,\mathcal{B}
\right],
$ where $\balpha^{(r)}$ denotes the parameter estimate at the $r^{\text{th}}$ iteration, $\mathcal{X}_o=\{\bx_{o,i}:i=1,\dots,n\}$ is the observed data and $\mathcal{B}=\{\mbpi:i=1,\dots,n\}$ is the feasible regions of the unobserved components.
Using \eqref{ll}, we obtain

\begin{align}
\label{Q function}
Q(\balpha)
&=
n\ln\Gamma(\alpha_0)
-
n\sum_{k=1}^p\ln\Gamma(\alpha_k)
+
\sum_{i=1}^n\sum_{k\in\mathcal{U}_i}\alpha_k\,\ev_{ik}
+
\sum_{i=1}^n\sum_{k\in\mathcal{O}_i}\alpha_k\ln x_{ik}
\nonumber\\
&\quad
-
\sum_{i=1}^n\sum_{k\in\mathcal{U}_i}\ev_{ik}
-
\sum_{i=1}^n\sum_{k\in\mathcal{O}_i}\ln x_{ik},
\end{align}
where $\ev_{ik}=
\mathbb{E}\!\left[
\ln X_{u,ik}
\mid
\bx_{o,i},\balpha^{(r)},\mbpi
\right],~ k\in\mathcal{U}_i.$ From Theorem~\ref{marginal_dist}, $\frac{\bx_{u,i}}{1-\|\bx_{o,i}\|_1}\ 
\Bigg|\ 
\bx_{o,i}
\sim
\bTD_{\mbmpci}(\balphau^{(r)}).$ Therefore
\begin{align}
\label{elnx}
\mathbb{E}[\ln X_{u,ik}\mid\bx_{o,i},\mbpi]
&=
\ln(1-\|\bx_{o,i}\|_1)
+
\mathbb{E}\!\left[
\ln\left(
\frac{ X_{u,ik}}{1-\|\bx_{o,i}\|_1}
\right)
\Bigg|
\bx_{o,i},\mbpi
\right].
\end{align}

Let
\begin{equation}
\bY_i=
\frac{\bx_{u,i}}{1-\|\bx_{o,i}\|_1}
\quad\text{and}\quad
Y_{ik}=
\frac{ X_{u,ik}}{1-\|\bx_{o,i}\|_1}.
\label{eq:Yi}
\end{equation}
Using the density in \eqref{truncated_pdf} of the truncated Dirichlet distribution, we obtain
\begin{align*}
\mathbb{E}\!\left[\ln Y_{ik}\mid \mbmpci\right]
&= \int_{\mbmpci}\ln y_{ik}\,
\frac{f_{\bTD_{\mbmpci}}(\byi;\balphau)}{F(\mbmpci;\balphau)}\,d\byi \\
&= \frac{1}{F(\mbmpci;\balphau)}
\bigintssss_{\mbmpci} 
\ln y_{ik}\,
\frac{\Gamma(\|\balphau\|_1)}{\displaystyle\prod_{j\in\us_i}\Gamma(\alpha_j)}
\prod_{j\in\us_i} y_{ij}^{\alpha_j-1}\, d\byi \\
&= \frac{1}{F(\mbmpci;\balphau)}
\frac{\partial}{\partial \alpha_k}
\bigintss_{\mbmpci}
\frac{\Gamma(\|\balphau\|_1)}{\displaystyle\prod_{j\in\us_i}\Gamma(\alpha_j)}
\prod_{j\in\us_i} y_{ij}^{\alpha_j-1}\, d\byi \\
& = \frac{\partial}{\partial\alpha_k}
\ln F(\mbmpci;\balphau)
+ \psi(\alpha_k)-\psi(\|\balphau\|_1),
\end{align*}
where the third equality follows from the identity $\partial y_{ik}^{\alpha_k-1}/\partial\alpha_k
= y_{ik}^{\alpha_k-1}\ln y_{ik}$, which allows the logarithmic term to be written as a derivative with respect to $\alpha_k$, and $\psi(\cdot)$ denotes the digamma function \citep{digamma}.
Substituting the expression for $\mathbb{E}[\ln Y_{ik}\mid \mbmpci]$ into \eqref{elnx} yields
\begin{align}
\label{elnx_r}
\mathbb{E}\!\left[
\ln
\left(
\frac{ X_{u,ik}}{1-\|\bx_{o,i}\|_1}
\right)\ 
\Bigg|\ 
\bx_{o,i},\mbpi
\right]
=
\frac{\partial}{\partial\alpha_k}
\ln F(\mbmpci;\balphau)
+
\psi(\alpha_k)
-
\psi(\|\balphau\|_1).
\end{align}
Hence, at iteration $r$,
\begin{align}
\ev_{ik}
=
\ln(1-\|\bx_{o,i}\|_1)
+
\frac{\partial}{\partial\alpha_k}
\ln F(\mbmpci;\balphau^{(r)})
+
\psi(\alpha_k^{(r)})
-
\psi(\|\balphau^{(r)}\|_1).
\end{align}
It is worth noting that the only term requiring numerical evaluation is 
$\frac{\partial}{\partial\alpha_k}\ln F(\mbmpci;\balphau^{(r)})$, 
which depends on the normalising constant of the truncated Dirichlet distribution. 
The remaining terms in \eqref{elnx_r} are available in closed form through the digamma function.

As discussed in Section~\ref{sec1}, missing-at-random (MAR) data can be viewed as a special case of unobserved data. 
Under a MAR mechanism, the feasible region becomes
$
\mbpi=
\{0\le x_{ik}\le1:\ k\in\mathcal{U}_i\},
$ in which case $F(\mbmpci;\balphau)=1$ and therefore $\ln F(\mbmpci;\balphau)=0$ for all $\balphau\in\mathbb{R}_+^{p_u}$.

\subsubsection{M-step}
\label{mstep}

The $k^{\text{th}}$ element of the gradient of \eqref{Q function} is
\begin{align}
\label{grad}
\frac{\partial Q(\balpha)}{\partial\alpha_k}
=
n\psi(\alpha_0)-n\psi(\alpha_k)
+
\sum_{i=1}^n \mathbb{I}_{\{k\in\mathcal{U}_i\}}\ev_{ik}
+
\sum_{i=1}^n \mathbb{I}_{\{k\in\mathcal{O}_i\}}\ln x_{ik},
\end{align}
for $k=1,\dots,p$, where $\mathbb{I}_A$ denotes the indicator function on the set $A$.
Maximisation of \eqref{Q function} requires solving
$
\frac{\partial Q(\balpha)}{\partial\alpha_k}=0$,
$k=1,\dots,p.
$
Such a maximisation can be performed numerically using the Newton--Raphson (NR) algorithm \citep{estimation_NR}, which is nested within the EM iterations. Since the NR method is unconstrained, the updates $\alpha_k^{(r+1)}$ may become negative. While a suitable initialisation (see Section~\ref{initialisation}) can mitigate this issue during the early iterations, positivity of the Dirichlet parameters must be guaranteed throughout the optimisation. 
To enforce this constraint while retaining an unconstrained search space, we adopt the reparameterisation proposed by \citet{estimation_transformation} $\alpha_k=e^{\beta_k}$, $\beta_k\in\mathbb{R},$ for $k=1,\dots,p$.
Under this transformation, we write the objective function as  $Q(\bm{\beta}) = Q(\balpha(\bm{\beta}))$. The gradient of $Q(\bm{\beta})$ with respect to $\beta_k$ is
\begin{align}
\label{grad_beta}
\nabla_k
%=\frac{\partial Q(\bm{\beta})}{\partial\beta_k}
=
e^{\beta_k}
\left[
n\psi\!\left(\sum_{j=1}^p e^{\beta_j}\right)
-
n\psi(e^{\beta_k})
+
\sum_{i=1}^n \ismiss\ev_{ik}
+
\sum_{i=1}^n \isobs\ln x_{ik}
\right].
\end{align}
Within the $(r+1)^{\text{th}}$ iteration of the EM algorithm, the NR update at the $s^{\text{th}}$ inner iteration is
\begin{align}
\label{nr_start}
\bm{\beta}^{(s+1)}
=
\bm{\beta}^{(s)}
-
(\bH^{(s)})^{-1}\bm{\nabla}^{(s)},
\end{align}
%where $\bm{\nabla}^{(s)} = \{\nabla_k^{(s)} : k=1,\dots,p\}$,
%$\bH^{(s)}$ is the Hessian matrix of $Q(\bm{\beta})$, and
%$(\bH^{(s)})^{-1}\bm{\nabla}^{(s)}$ is the search direction,
%all evaluated at $\bm{\beta}^{(s)}$.

where $\bm{\nabla}^{(s)} = \{\nabla_k^{(s)} : k=1,\dots,p\}$,
$\bm{H}^{(s)}$ is the Hessian matrix of $Q(\bm{\beta})$, and
$(\bm{H}^{(s)})^{-1}\bm{\nabla}^{(s)}$ is the search direction,
all evaluated at $\bm{\beta}^{(s)}$. The second derivatives of $Q(\bm{\beta})$ are
\begin{align*}
\frac{\partial^2 Q(\bm{\beta})}{\partial\beta_j\partial\beta_k}
=
\begin{cases}
n e^{\beta_j+\beta_k}\psi'\!\left(\displaystyle\sum_{l=1}^p e^{\beta_l}\right),
& j\neq k,\\[6pt]
\nabla_k
+
n e^{2\beta_k}\psi'\!\left(\displaystyle\sum_{l=1}^p e^{\beta_l}\right)
-
n e^{2\beta_k}\psi'(e^{\beta_k}),
& j=k.
\end{cases}
\end{align*}
where $\psi'(\cdot)$ denotes the trigamma function. Define $\gamma^{(s)} = n\psi'\!\left(\sum_{k=1}^p e^{\beta_k^{(s)}}\right)$, 
$\ebv^{(s)} =  \begin{bmatrix} e^{\beta_1^{(s)}} & \dots & e^{\beta_p^{(s)}} \end{bmatrix}^{\top},$ and $\bm{D}^{(s)} = \mathrm{diag}\!\left(
\nabla_1^{(s)}-n e^{2\beta_1^{(s)}}\psi'(e^{\beta_1^{(s)}}),
\dots,
\nabla_p^{(s)}-n e^{2\beta_p^{(s)}}\psi'(e^{\beta_p^{(s)}})
\right).$

Then, the Hessian can be written as $\bH^{(s)}
=
\bm{D}^{(s)}
+
\gamma^{(s)}\ebv^{(s)}\ebv^{(s)\top}.
$
Using the Sherman--Morrison formula \citep{sherman1950adjustment}, the inverse Hessian is
\begin{align}
\label{sm_form}
(\bH^{(s)})^{-1}
=
(\bm{D}^{(s)})^{-1}
-
\frac{\gamma^{(s)}(\bm{D}^{(s)})^{-1}\ebv^{(s)}\ebv^{(s)\top}(\bm{D}^{(s)})^{-1}}
{1+\gamma^{(s)}\ebv^{(s)\top}(\bm{D}^{(s)})^{-1}\ebv^{(s)}}.
\end{align}
Since $\bm{D}^{(s)}$ is diagonal, letting $d_{kk}^{(s)}$ denote its $k^{\text{th}}$ diagonal element, the NR update simplifies to $\beta_k^{(s+1)} = \beta_k^{(s)} - \frac{\nabla_k^{(s)}-w^{(s)}e^{\beta_k^{(s)}}}{d_{kk}^{(s)}},$
where $
w^{(s)}
=
\frac{\displaystyle
\sum_{k=1}^p
\frac{e^{\beta_k^{(s)}}\nabla_k^{(s)}}{d_{kk}^{(s)}}}
{\displaystyle
\frac{1}{\gamma^{(s)}}
+
\sum_{k=1}^p
\frac{e^{2\beta_k^{(s)}}}{d_{kk}^{(s)}} } .$ The NR iterations continue until convergence, yielding 
$\bm{\beta}^{(r+1)}$.
The EM update of $\alpha_k$ is therefore
$
\alpha_k^{(r+1)} = e^{\beta_k^{(r+1)}}$,
$k=1,\dots,p.$
%Initialisation and convergence criteria are discussed in Section \ref{initialisation} under supplementary material.
Initialisation is discussed in Section \ref{initialisation} under supplementary material.

%%%%%%NEW PART%%%%
Two stopping rules must be specified for the proposed algorithm. 
An inner stopping rule controls the convergence of the Newton--Raphson (NR) procedure used within the M-step, while an outer stopping rule determines the convergence of the overall EM algorithm.
For the NR algorithm, convergence is monitored using the quantity $\frac{1}{2}\lambda^2$  where:
$ \lambda  = \sqrt{ \left(\Delta \bm \beta^{(s+1)}  \right)^{\top} \bm H^{(s+1)} \Delta \bm\beta^{(s+1)}    },
$
with $\Delta \bm \beta^{(s+1)} = \bm \beta ^{(s+1)} - \bm \beta ^{(s)}$ and $\bm H^{(s+1)}$ is the Hessian matrix evaluated at $\bm \beta ^{(s+1)}$. The criterion, $\frac{1}{2}\lambda^2$ provides an approximation to the distance from the current objective function value to its supremum, namely $\left| Q(\bm{\beta}) - \underset{\bm\beta}{\mathrm{sup}}Q(\bm\beta)\right|$. Thus, iterations are stopped once $\frac{1}{2}\lambda^2$  falls below a small tolerance.

%The observed-data log-likelihood increases monotonically under the EM algorithm. However, temporary stability may occur when the algorithm approaches a local maximum before moving towards the global maximum.
To assess convergence to the asymptotic log-likelihood, we employ the Aitken acceleration criterion. The complete-data likelihood involves unobserved variables and is therefore not directly usable for this purpose; consequently, convergence is monitored using the observed-data log-likelihood based on the marginal distribution of the observed data given in formula \eqref{marginal}. Let $l_o^{(r)}$ denote the observed log-likelihood at the $r^{\text{th}}$ EM iteration. The Aitken acceleration  and the corresponding estimate of the asymptotic log-likelihood are defined as
\begin{align*}
a^{(r+1)} =
\frac{l_o^{(r+2)}-l_o^{(r+1)}}{l_o^{(r+1)}-l_o^{(r)}}, \qquad
(l_o^{\infty})^{(r)} =
l_o^{(r+1)} +
\frac{l_o^{(r+2)}-l_o^{(r+1)}}{1-a^{(r+1)}}.
\end{align*}
The EM algorithm is considered to have converged when
$(l_o^{\infty})^{(r)}-l_o^{(r+1)} < \epsilon,
$
where $\epsilon>0$ is a small tolerance \citep{convergence_aitken}.

% \subsection{Finite mixtures of Dirichlet distributions}

% The algorithm outlined in Section \ref{ecm} can be extended to accommodate heterogeneous compositional data with missing values. 
% When the observed proportions arise from multiple underlying populations, a finite mixture of Dirichlet distributions provides a flexible modelling framework. 
% The mixture density is given by
% \begin{align}
%     f_{M\bD}(\bx;\bpsi) = \sum_{g=1}^G \pi_g f_{\bD}(\bx;\balpha_g),
%     \label{fmm}
% \end{align}
% where $\bpsi = \{\pi_g,\balpha_g : g=1,\dots,G\}$. 
% The mixing proportions satisfy $\pi_g>0$ and $\sum_{g=1}^G \pi_g = 1$, and $f_{\bD}(\cdot)$ is defined in \eqref{pdf_dirichlet}.

\section{Application: Mercury levels in blood}
\label{application}
The application demonstrates the usefulness of the proposed algorithm to gain for gaining insight into compositional data and its effectiveness in imputing unobserved values in incomplete datasets. Specifically, we consider a dataset on mercury levels and species in blood samples collected through surveys. The dataset contains both censored and missing values.

The proposed algorithm is applied to data from the National Health and Nutrition Examination Survey (NHANES) survey programme. Its public repository contains, among other data, laboratory results from volunteers’ blood samples. We focus on measurements of mercury species. In addition to total mercury, NHANES measures inorganic and organic mercury. Inorganic mercury typically enters the body through environmental exposure, including skin-care products. Organic mercury is characterised by includes two types: methylmercury and ethylmercury. Total mercury and its species are measured independently. Thus, although the total amount may be determined, the amount of a specific mercury type may remain unobserved. Uncertainty remains regarding methylmercury exposure from fish consumption and its potential health effects. Previous exposure estimates have relied on food-consumption surveys and measurements of methylmercury in fish. Biomarkers of exposure are therefore needed to improve exposure assessment.

Methods used to detect and quantify mercury species are subject to lower limits of detection (LLOD). Since only two types of organic mercury are record recorded, other unidentified mercury types may also be present in blood. The limits are given in \tablename~\ref{llod}.

\figurename~\ref{pair_plots} in the Supplementary section suggests that proportions of mercury types are homogeneous and negatively correlated. A Dirichlet distribution is thus fitted to the data, consisting of $n=3872$ rows with a total of 71.978\% of cells unobserved. The estimated parameters and moments are given in \tablename~\ref{mle}. The estimated parameters and moments are reported in \tablename~\ref{mle}. The proportion of methylmercury is the highest. This result is consistent with known sources of mercury exposure. In particular, methylmercury is readily absorbed and bioaccumulates through the food chain, which may explain why it accounts for the largest proportion among the mercury types observed in the dataset. Inorganic mercury has the second-largest average proportion. Typical sources of inorganic mercury exposure include environmental sources such as paints, lotions, electronics, and related products. Ethylmercury has the smallest average proportion. This form of mercury has a shorter biological half-life than methylmercury, and exposure to ethylmercury generally occurs in more specific contexts, such as medicinal injections. Thus, it is expected that this type of mercury would represent a smaller proportion of total mercury exposure.
An additional benefit of focusing on mercury proportions is that this approach highlights the presence of mercury types that are not directly identified in the data, and hence the possibility of additional, unknown sources of mercury exposure. These unknown types appear to be prevalent enough to emerge as a gap in blood-test measurements. Moreover, the proportions of these unknown mercury types are themselves largely unknown, since their estimation depends on the other measured quantities, which may also be subject to censoring. The algorithm used here provides a way to estimate the proportions of unknown mercury types while accounting for the key characteristics of the data, namely their positive support and additive constraint.

We next compare the imputations with those from competing methods. Comparison with the ALR-EM approach is difficult because some rows contain only one observed part, whereas ALR-EM requires at least two. Removing incompatible rows yields a subsample of $n=462$. Within this subsample, residual mercury and ethylmercury have only 4 and 5 observed values, respectively, making reliable estimation of the ALR-EM scale matrix difficult. To address this, four unobserved ethylmercury values are replaced by the conventional value of one-half of the detection limit, providing enough observations for ALR-EM. Because the unobserved values are genuinely unavailable, imputation accuracy cannot be evaluated by direct comparison with the true values. Instead, we adopt the heuristic that the model achieving the best fit to the observed data should provide the most reliable imputations. This is particularly appropriate here, since imputations under both the proposed method and the ALR-EM algorithm are based on conditional expectations. Accordingly, model performance can be assessed using criteria derived from the observed-data likelihood. Evaluation of this likelihood requires the marginal distribution of the observed components; however, for the multivariate logistic-normal distribution, this marginal distribution is not available in closed form. This limits the direct use of standard likelihood-based fit measures. To overcome this difficulty, we exploit the fact that any closed subcomposition of a multivariate logistic-normal random vector is itself multivariate logistic normal \citep{aitchison_logistic}. The Dirichlet distribution's neutrality property ensures the same relationship between a random vector and its closed subset \cite{dirichlet_neutrality}. Thus, after fitting each approach to the subsample, we compute the usual model-selection measures on the closed observed subcompositions. Note that the Dirichlet distribution is fitted to the subsample without substituted values. The model-selection metrics are reported in \tablename~\ref{gof}. For ease of interpretation, the Akaike information criterion (AIC) and Bayesian information criterion (BIC) are reported as $2l_o -2\rho$ and $2l_o - \rho \ln(n)$, respectively, where $l_o $ is the observed log-likelihood value at convergence and $\rho$ denotes the number of free parameters. Under this convention, larger values indicate a better fit. The Dirichlet model attains the largest observed log-likelihood, AIC, and BIC. This suggests that the Dirichlet model provides a more plausible representation of the observed data and may therefore yield more reliable imputations.

The metrics in \tablename~\ref{gof} focus on the observed components rather than the full compositional structure. To address this limitation, we note that the Dirichlet and multivariate logistic-normal distributions are not identical for any parameter values, although they may sometimes be approximately equal. If the fitted distributions differ substantially, they represent the data in distinct ways, and the relevant question becomes which is closer to the underlying distribution. We measure the difference between the fitted models using the Wasserstein distance and compare it with the distance obtained when the distributions are identical. Roughly speaking, the Wasserstein distance is a symmetric measure of the minimal effort required to reconfigure the probability mass of one distribution in order to recover the other distribution \cite{wasserstein}. Since the distributions are defined on the simplex, the distance measure used in the Wasserstein calculation is the Aitchison distance given in \eqref{aitchison_distance}. The matrix of distances are reported in \tablename~\ref{wasserstein}. Notice that the minimal amount of work needed to transform the probability mass of the multivariate logistic normal distribution into that of the Dirichlet is around 14 units, taking into account the Simplex geometry. This metric implies that the fitted Dirichlet and multivariate logistic normal distributions model the dataset differently. Combining this result with the model selection metrics in \tablename~\ref{gof}, it is plausible to consider the Dirichlet distribution a better fit.  

Comparing either the proposed algorithm or the ALR-EM approach with geometric mean imputation presents another challenge, since the latter is non-parametric. In this case, we compare the marginal densities of the imputations obtained from the three approaches. Geometric mean imputation rescales the observed components so that the completed compositions satisfy the additive constraint. That is, the observed parts are rescaled by this imputation technique, which may alter their marginal distributions. It is therefore sufficient to assess whether this warping is detrimental to the distribution of the sample, while also examining the distribution of the imputed values relative to those obtained from the model-based counterparts. From the histograms in \figurename~\ref{histograms}, the values imputed by geometric mean imputation appear to follow a distribution that is distinct from that of the other methods. 

\section{Conclusion} 
Incomplete compositional data analysis faces a fundamental limitation: although the Dirichlet distribution is the natural probabilistic model for compositions, no likelihood-based estimation method directly accommodates incomplete proportions on the unit simplex. Consequently, analysts are typically forced either to discard partially observed compositions or to transform the data into unconstrained spaces, thereby sacrificing interpretability and coherence with the compositional structure \citep{alr_em_algorithm}.
This paper addresses this gap by developing a model-based estimation framework for Dirichlet models in the presence of missingness and censoring, which can both be seen as a form of coarsening. An EM-type algorithm is proposed to obtain maximum likelihood estimates of the Dirichlet distribution from incomplete compositional data while remaining entirely within the unit simplex. The method accommodates two coarsening mechanisms: coarsening completely at random (CCAR) and coarsening at random (CAR).

Simulation studies demonstrate that the proposed method substantially improves parameter recovery compared with maximum likelihood estimation based only on complete cases. Moreover, imputations produced by the proposed algorithm exhibit smaller errors than those obtained using the current model-based approach and geometric mean imputation, which is commonly used for its simplicity and minimal assumptions. Notably, the algorithm performs reliably even under extreme conditions, including observations with only a single observed component and datasets with up to 90\% unobserved values, highlighting its reliability and broad applicability.

More broadly, this work reinforces the role of the Dirichlet distribution and its extensions as natural models for compositional data. When incomplete compositions arise, analysts have historically relied on methods designed for unconstrained data, adapting the data rather than the model. The approach proposed here reverses this perspective by developing estimation tools that respect the geometry and constraints of compositional data. By enabling likelihood-based inference and clustering directly on the simplex in the presence of unobserved values, this work opens the door to further developments in simplex-based modelling, including richer mixture structures and more flexible inferential procedures.

\label{conclusion}

\section{Disclosure statement}\label{disclosure-statement}
The authors declare no conflict of interest exits.
%The authors have the following conflicts of interest to declare (or replace with a statement that no conflicts of interest exist).

\section{Data Availability Statement}\label{data-availability-statement}

Deidentified data have been made available at: \url{https://wwwn.cdc.gov/nchs/nhanes/search/datapage.aspx?Component=Laboratory&Cycle=2021-2023}.

\bibliography{bibliography.bib}

\phantomsection\label{supplementary-material}
\bigskip

\newpage
\begin{center}

{\large\bf SUPPLEMENTARY MATERIAL}

\end{center}

\subsection{Initialisation}
\label{initialisation}

The Newton--Raphson (NR) algorithm, performed within each M-step, requires suitable starting values to improve the chances of convergence to the global maximum of the log-likelihood.
When the NR algorithm is applied directly to the Dirichlet parameters $\balpha$, 
it has been noted that setting each $\alpha_k$ equal to the minimum value of the 
$k^{\text{th}}$ data column prevents the algorithm from producing negative updates 
when the method-of-moments estimates are close to zero \citep{initialisation_roning}.
Due to the $\bbeta$-reparameterisation adopted in the M-step, the issue of negative values no longer arises. 
Consequently, the method-of-moments estimates provide convenient and reliable starting values \citep{estimation_NR}. 
The moment-based initialisation of the Dirichlet parameters only requires the first two raw moments of each component. Let $\mathcal{O}_k = \{i : x_{ik} \text{ is observed}\}$ denote the set of observations for which the $k^{\text{th}}$ component is available, and let $n_k = |\mathcal{O}_k|$. 
The $j^{\text{th}}$ raw sample moment of the $k^{\text{th}}$ component is computed using the available observations as
\begin{align}
\overline{x}_{jk}
=
\frac{1}{n_k}
\sum_{i \in \mathcal{O}_k} x_{ik}^{\,j}.
%,
%\qquad j=1,2 .
\end{align}
As documented by \citet[][Section~2.8.1]{ng2011dirichlet}, the method-of-moments estimate of the $k^{\text{th}}$ element of $\balpha$ is given by
\begin{align}
\overline{\alpha}_k =
\frac{\overline{x}_{11}-\overline{x}_{21}}
{\overline{x}_{21}-\overline{x}_{11}^2}
\,\overline{x}_{1k},
\qquad k=1,\dots,p .
\end{align}

\section{Simulation experiment}
\label{simulations}
The goal of the simulation experiment is to evaluate two aspects of the
proposed algorithm: (1) parameter estimation for the Dirichlet distribution in
the presence of incomplete data, and (2) imputation of the unobserved
components. The incomplete-data settings considered here include missing data
and censored data. Moreover, the three coarsening mechanisms are included: CCAR, CAR, and CNAR. CNAR was included to test the method's performance when the assumptions are violated.
Accordingly, the simulation design examines the performance of the algorithm’s estimators when the data contain missing CCAR (MCAR) values, missing CAR (MAR) values, and censored CNAR observations. The study also varies the sample size, thereby providing evidence regarding the large-sample properties of the estimation procedure. In addition, it evaluates the quality of the imputations produced by the algorithm.
The simulation experiment is divided into three parts. In parts A, and B, four different sample sizes are considered:  $n=50$, $n = 200$, $n=500$, and $n=1000$. In part C, four sample sizes are unchanged except $n=50$ is changed to $n=100$. For each sample size, $1000$ datasets are simulated.

\begin{itemize}
    \item In Part A, for each dataset, a percentage of observations is hidden completely at random. The proposed algorithm is then fitted, and its performance is evaluated in terms of how accurately it recovers the true parameters. The results are compared with those obtained by fitting a Dirichlet distribution using only the fully observed rows. Part A concludes by assessing the imputation accuracy of the proposed method, comparing it with a current model-based imputation method and a nonparametric geometric mean imputation approach.

    \item Part B examines the performance of the algorithm using the same evaluation criteria as in Part A, but with data hidden under CAR mechanism.

    \item In Part C, parameter recovery and imputation is assessed using CNAR censored data. For the imputation, left-censored data are generated to maintain comparability with the model-based competing method.
\end{itemize}

To evaluate how well the fitted distributions recover the true parameters, we use the average absolute bias (AAB) and root mean squared error (RMSE) metrics, as defined in \citet{test}. 
Let $\btheta$ be the unknown parameter of dimension $p$ to be estimated.
For the $b^{\text{th}}$ dataset, let the estimated values be denoted as $\hat{\btheta}^{(b)}$, for $b = 1, \dots, 1000$. The average absolute bias and RMSE are defined as:
\begin{align}
\label{ab_and_rmse}
\text{AAB}(\hat{\btheta}) &= \frac{1}{1000} \sum_{b=1}^{1000} \sum_{k=1}^p |\hat{\theta}^{(b)}_{k} - \theta_{k}| \hspace{0.5cm} \text{ and } \hspace{0.5cm} \text{RMSE}(\hat{\btheta}) = \sqrt{ \frac{1}{1000} \sum_{b=1}^{1000} \sum_{k=1}^p \left(\hat{\theta}^{(b)}_{k} - \theta_{k} \right)^2 }.
\end{align}
Additionally, the proposed algorithm provides imputations for the missing elements in the simulated datasets. We therefore assess how closely these imputations match the true simulated values. Since each observation lies on the simplex, Euclidean distance is not an appropriate measure of discrepancy. Instead, the Aitchison distance is used to compare the imputed values with the corresponding true values. Denoting the $i^{\text{th}}$ true vector as $\bx_i$ and the imputed vector as $\bx_i^*$, the squared Aitchison distance between the two is given as \citep{aitchison_distance}:
\begin{align}
\label{aitchison_distance}
    d^2_i(\bx_i, \bx^*_i) = \left [ \sum_{k=1}^p \left(\ln\frac{x_{ik}}{g(\bx_{i})} - \ln\frac{x^*_{ik}}{g(\bx^*_{i})} \right)\right ]^2,
\end{align}
where $g(\bx) = \left [\displaystyle\prod_{k=1}^px_k \right]^{\frac{1}{p} }$ is the geometric mean. The average compositional error variance measure is defined as \citep{alr_em_method}:
\begin{align}
    CEV = \frac{1}{1000} \sum_{b=1}^{1000} \frac{1}{m_{b} } \sum_{i=1}^{ m_{b} }  d^2_i(\bx_i, \bx^*_i),
\end{align}
where $m_b$ is denotes the number of rows with at least two unobserved cells in the $b^{\text{th}}$ simulated dataset.

\subsection{Simulation results: Part A}
Datasets are simulated from a Dirichlet distribution with the parameter vector:
\begin{align}
\label{true_parm}
 \balpha = [ 6, 3, 5, 0.5,2]^{\top}.
\end{align}
The chosen parameter values generate data with different concentration levels, resulting in a variety of contour shapes. This choice provides a representative view of the different behaviours that the Dirichlet density in \eqref{pdf_dirichlet} can exhibit, shown in \figurename ~\ref{contours}.
\begin{figure}[H]
    \centering
    \fbox{ \includegraphics[width=0.65\textwidth]{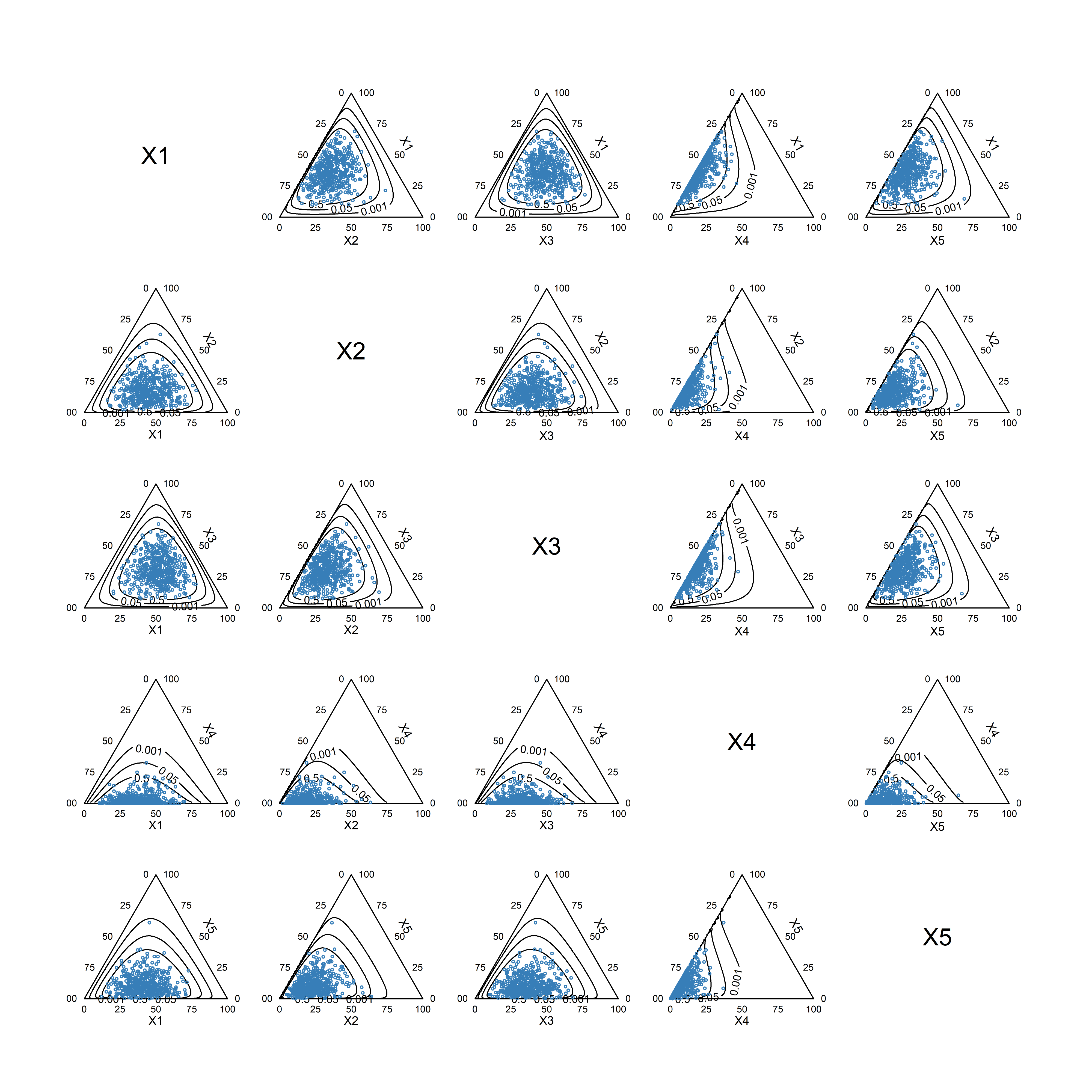}}
    \caption{Pairwise contour plots of a Dirichlet distribution with parameter vector given by \eqref{true_parm}. }
    \label{contours}
\end{figure}
The observations for random variables $[X_1, X_2, X_3, X_4]^{\top}$ are hidden according to the percentages 0\%, 10\%, 20\%, $\dots$, and 90\%. 

\subsubsection{Simulation results: Part A - Parameter recovery}
\begin{figure}[H]
    \centering
    \includegraphics[width=\textwidth]{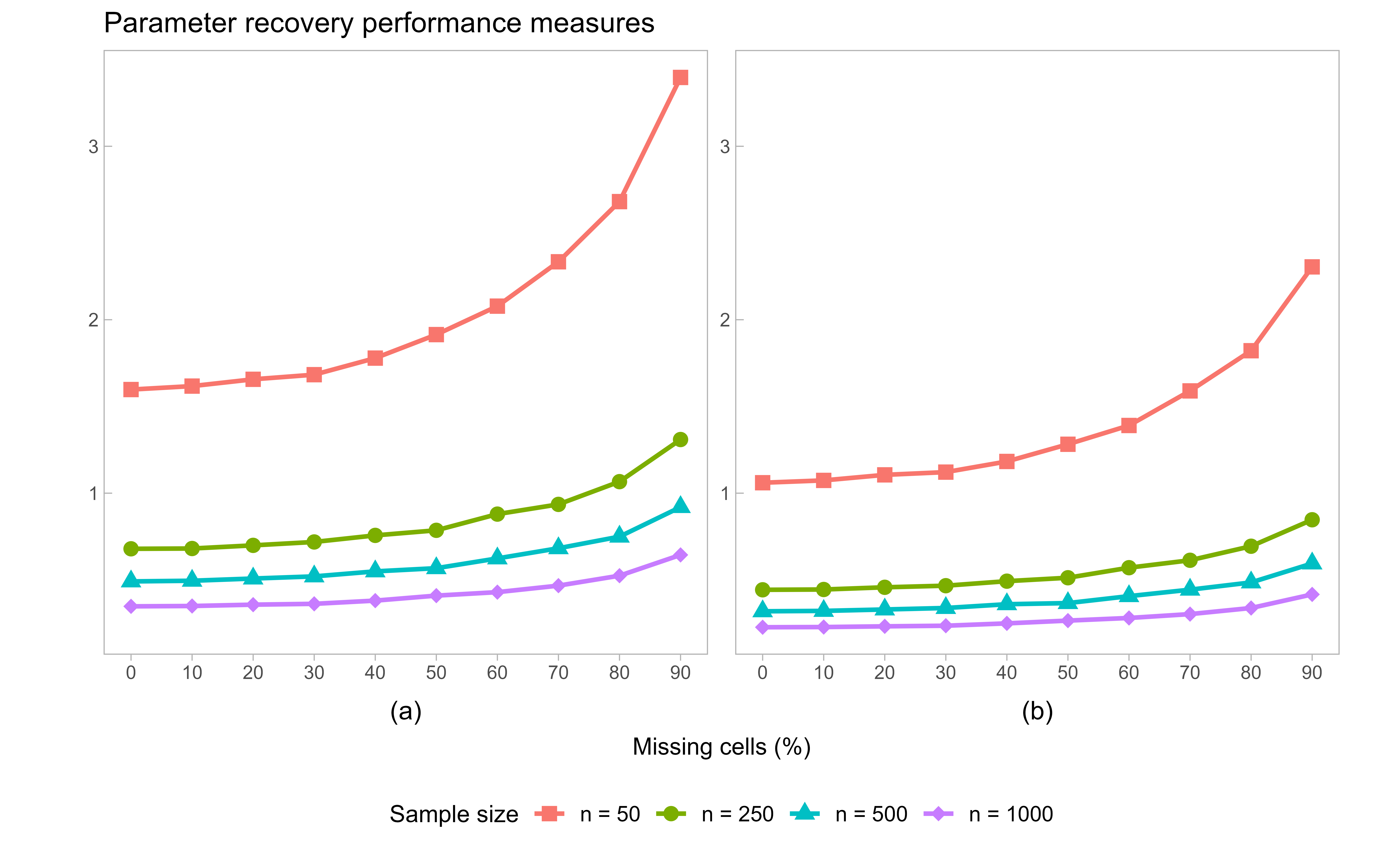}
    \caption{(a) AAB for percentage of missing values and (b) RMSE for percentage of missing values for simulation Part A. }
    \label{alpha_measures}
\end{figure}

\begin{figure}[H]
    \centering
    \includegraphics[width=\textwidth]{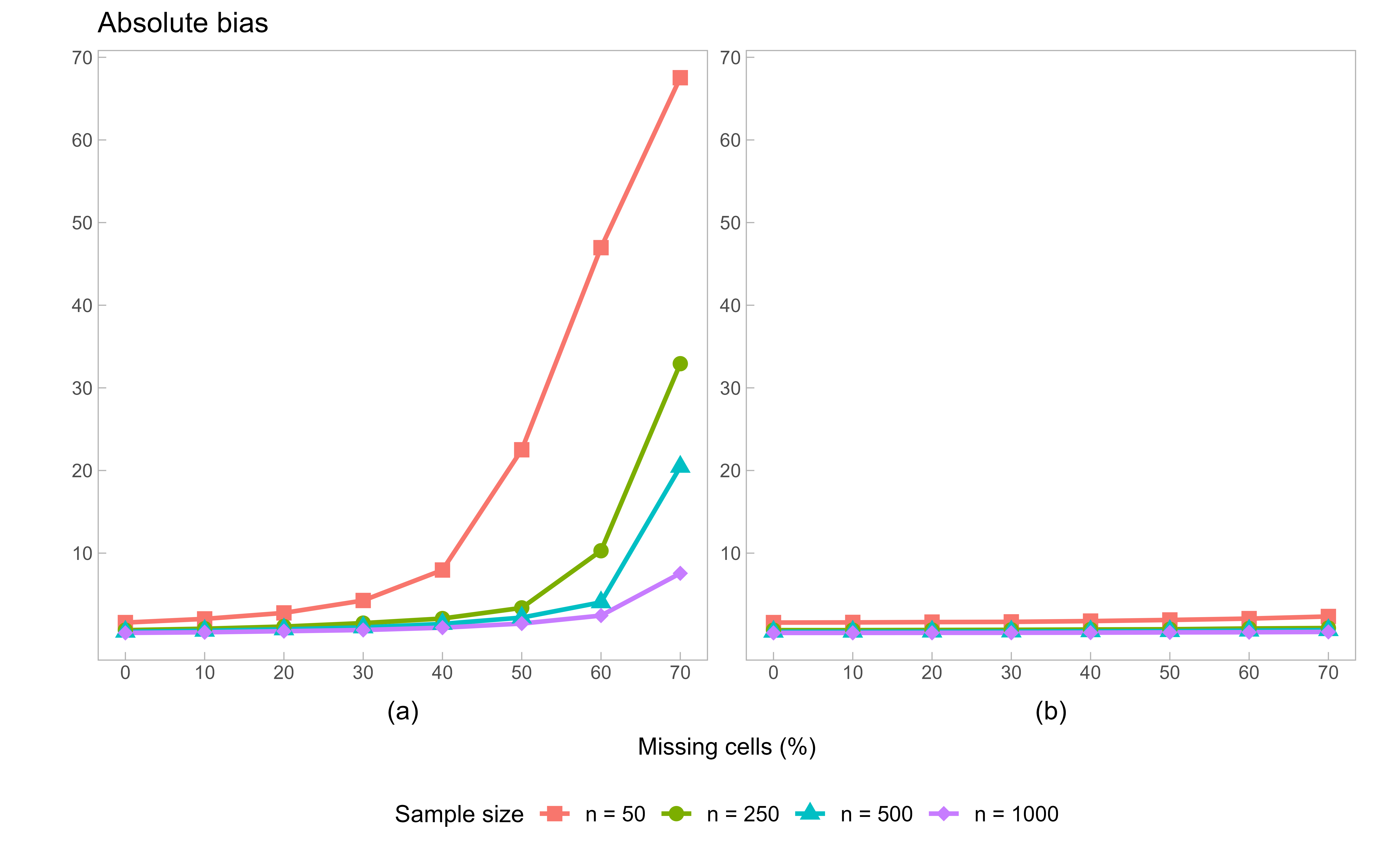}
    \caption{(a) AAB of parameter estimates obtained using only non-missing rows of the simulated dataset and (b) AAB using the proposed algorithm  for simulation Part A. }
    \label{alpha_bias_comparison}
\end{figure}

\begin{figure}[H]
    \centering
    \includegraphics[width=\textwidth]{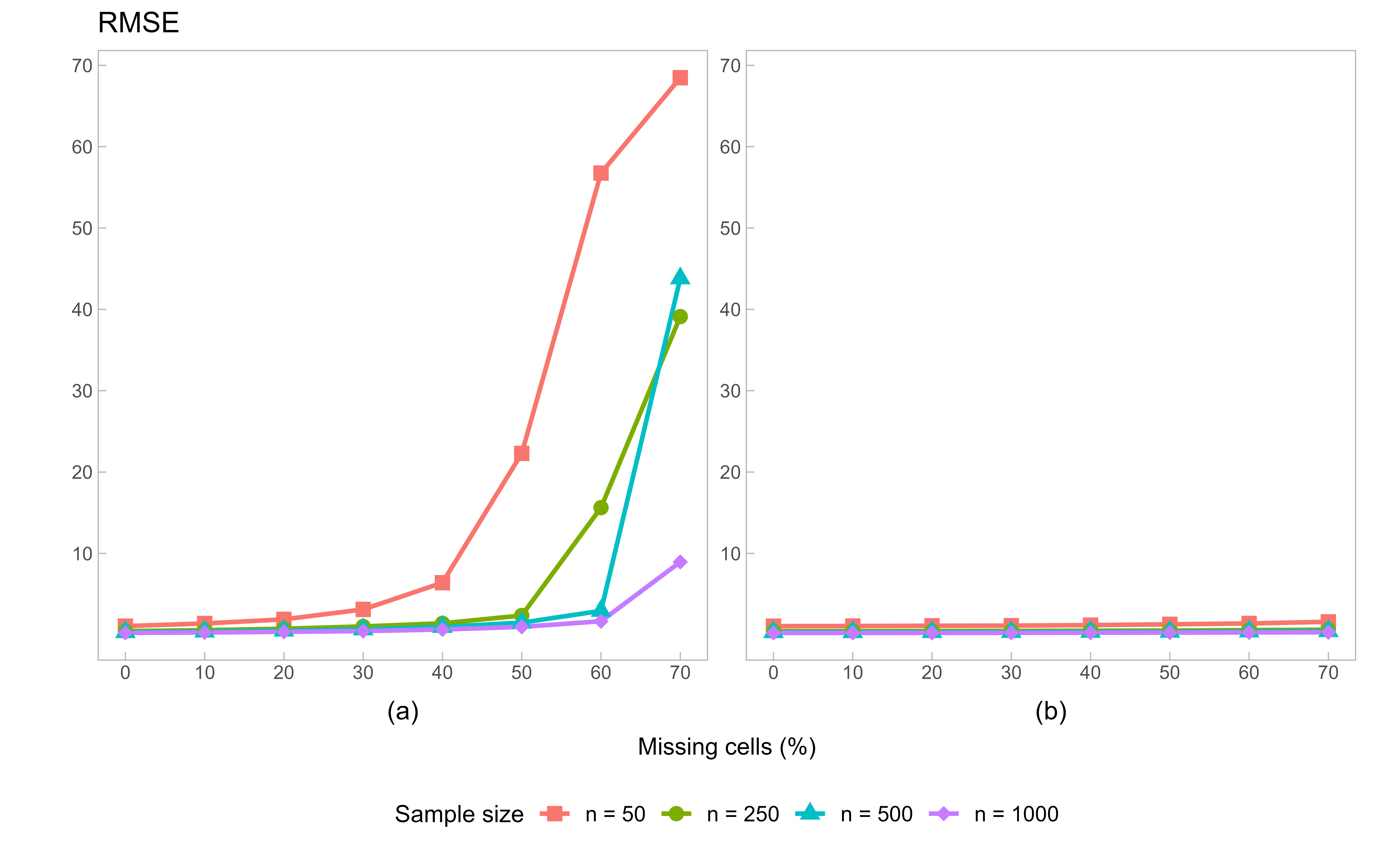}
    \caption{(a) RMSE of parameter estimates using only non-missing rows of the simulated dataset and (b) RMSE using the proposed algorithm for simulation Part A. }
    \label{alpha_rmse_comparison}
\end{figure}
As expected, both AAB and the variability of the estimated parameters increase as the proportion of missing values in the dataset grows, as shown in Figure \ref{alpha_measures}. This occurs because fewer observed values imply less information is available, and the estimation procedure must rely increasingly on the imputations produced by the algorithm. This also explains the reduction in error metrics as the sample size increases (from 
$n=50$ to $n=1000$).
Although this degradation in performance is an inherent consequence of working with incomplete data, the results demonstrate that reliable parameter estimates can still be obtained even when the dataset is heavily affected by missing values. In particular, the maximum observed values of AAB and RMSE are approximately 3 and 2, respectively, and occur when 90\% of the data are missing completely at random.
Figures \ref{alpha_bias_comparison} and \ref{alpha_rmse_comparison} show that these error levels represent a substantial improvement compared with the approach that fits the Dirichlet distribution using only the fully observed rows. For this comparison, the missingness proportion was capped at 70\%, since NR procedure failed to converge reliably at higher levels of missingness. Notably, for small sample sizes ($n=50$), the algorithm reduces the estimation error by a factor greater than 20.
Consequently, the advantages of the proposed algorithm are twofold: first, parameter estimation remains possible even when the dataset contains no complete observations; second, the resulting estimates provide substantially better recovery of the true parameters than those obtained using row deletion. The AAB and RMSE grow extremely large for higher percentages of missing values. This is because the parameters are estimated on a subset of values with lowered variation. As a consequence, the maximisation of the likelihood function causes the density to concentrate around the complete observations and it results in the estimator of the parameter $\balpha$ to grow arbitrarily large.

\subsubsection{Simulation results: Part A - Imputation}
To demonstrate the proposed algorithm’s imputations performance are compared with those from two popular existing techniques. The first technique applies an additive logratio transform (ALR) to a compositional dataset in order to remove the additive constraint on the rows. The transformed data then lie in $\mathbb{R}_p$ rather than on the simplex $\simplex$, and are assumed to follow a multivariate normal distribution. The EM algorithm for fitting a multivariate normal distribution to incomplete data is well established \citep{alr_em_algorithm, alr_em_method, alr_in_r}; therefore, missing values in the transformed space are imputed using this EM procedure. After convergence, the resulting imputations are mapped back onto $\simplex$ using the inverse additive logratio transformation. This approach is referred to as a modified logratio EM algorithm (ALR-EM).
The second technique considered is geometric mean imputation. Its popularity stems primarily from its simplicity: no transformation of the data is required, and it constitutes a nonparametric approach \citep{geo_mean}.
For a consistent comparison between these two competing approaches and the proposed algorithm, observations for the random variables $[X_1, X_2, X_3]^{\top}$ are hidden according to the percentages 0\%, 10\%, 20\%, 30\%, $\dots$, 70\%. This design is necessary because the ALR transformation requires at least two observed components in each row to be defined. Furthermore, for small sample sizes, the initialisation of the scale matrix can become singular when the percentage of missing values exceeds 70\%, preventing the ALR-EM algorithm from running and thereby interrupting the simulation. Thus, no results are shown for missingx percentages above 70\%.
\begin{figure}[H]
    \centering
    \includegraphics[width=0.6\textwidth]{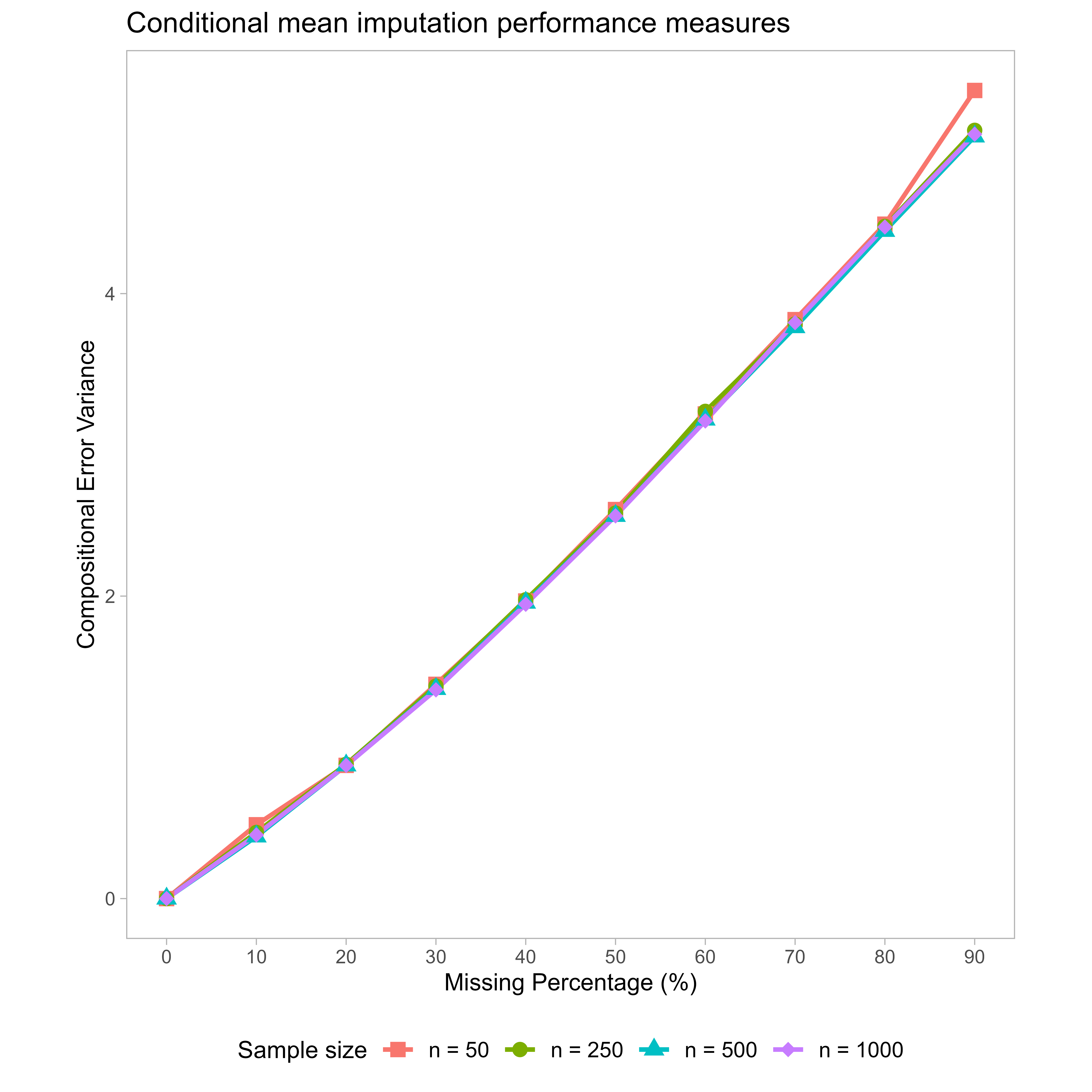}
    \caption{CEV for imputed values after fitting a Dirichlet distribution with the proposed algorithm under varying percentage of missing values in the first four columns of simulated datasets. }
    \label{cev}
\end{figure}

\begin{figure}[H]
    \centering
    \includegraphics[trim={1.8cm 0cm 1cm 1cm}, clip, width=1\linewidth]{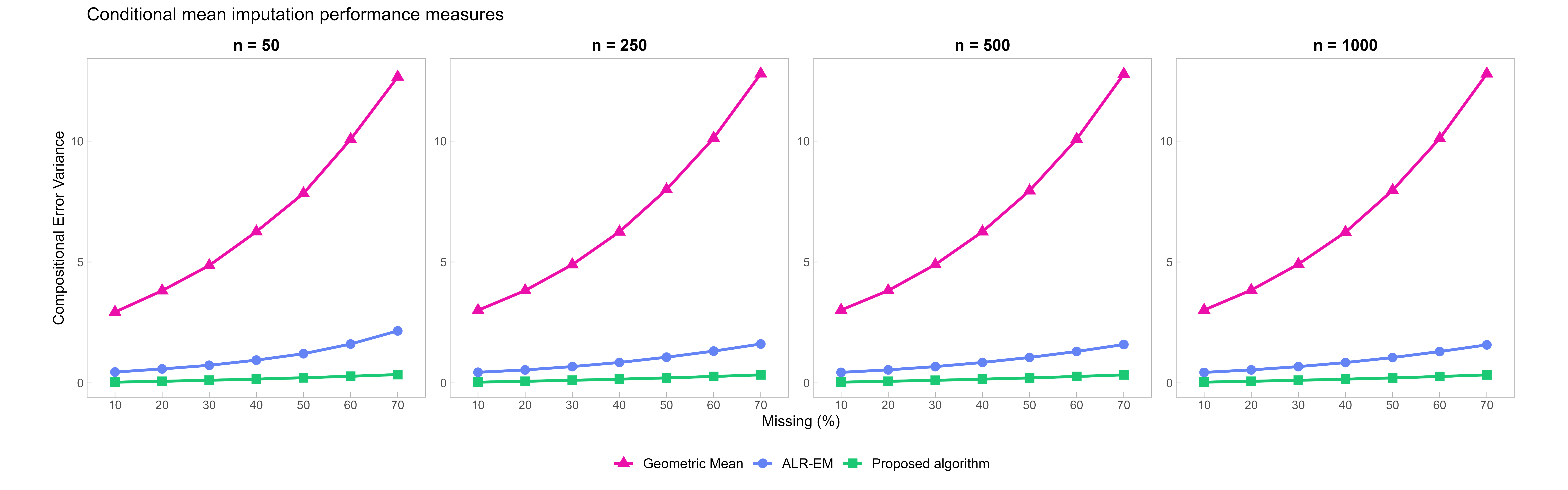}
    \caption{CEV for imputed values from geometric mean imputation, ALR-EM, and the proposed algorithm under varying percentages of missing elements in the first three columns of the simulated datasets.}
    \label{cev_comparison}
\end{figure}
Interestingly, from Figure \ref{cev}, we see that the difference between observed and imputed values grows linearly as the percentage of missing values increases. The error is also independent of sample size. This is because the number of imputations grows with the sample size. 

The CEVs for each method, plotted in Figure \ref{cev_comparison}, show that imputations produced by the proposed algorithm are more similar to the true values than those obtained from the competing methods. In particular, the geometric mean imputation exhibits the poorest performance relative to the ALR-EM algorithm, which may be attributed, in part, to the fact that this method ignores the covariance structure inherent in compositional data.
These results suggest that the transform and back-transform steps in the ALR-EM algorithm introduce additional differences.   As a result, dissimilarity increases more rapidly compared with the proposed algorithm, which shows a more stable, approximate linear trend.
As with the results shown in Figure \ref{cev}, the sample size does not appear to substantially affect the outcomes. Instead, the results highlight the flexibility of the proposed algorithm: it can handle datasets more severely affected by missingness and requires only one observed component per row, whereas the competing methods require at least two observed components.

\subsection{Simulation results: Part B}
In this part, for each of the 1000 datasets generated, observations for random variables $[X_1, X_2, X_3, X_4]^{\top}$ are hidden depending on the observed values of $X_5$. The assumption is that the larger values of $X_5$ are associated with higher likelihoods of missing values in $[X_1, X_2, X_3, X_4]^{\top}$. Denoting $q_i$ as the $i^{\text{th}}$ empirical quartile of $X_5$, the probability of values missing in $[X_2, X_3, X_4]^{\top}$ are: 
\begin{flalign}
\label{mar_scheme}
&\mathbb{P}(X_k \text{ is missing }|x_5) = \begin{cases}
    0.1 & \text { if } 0   \leq x_5 < q_1, \\
    0.2 & \text { if } q_1 \leq x_5 < q_2, \\ 
    0.5 & \text { if } q_2 \leq x_5 < q_3, \\ 
    0.6 & \text { if } q_3 \leq x_5,   
\end{cases} &
\end{flalign}
for $k = 2,3,4$. Furthermore, $X_1$ is subjected to a different set of probabilities:\\
$\mathbb{P}(X_1 \text{ is missing }|x_5) = \begin{cases}
    0.1 & \text { if } 0   \leq x_5 < q_1, \\
    0.2 & \text { if } q_1 \leq x_5 < q_2, \\ 
    0.4 & \text { if } q_2 \leq x_5 < q_3, \\ 
    0.9 & \text { if } q_3 \leq x_5.   
\end{cases} $.\\
This follows the design of \citet{mar_simulation_setup} for evaluating models on MAR data. Since missingness is governed by a probabilistic mechanism, the percentage of missing values is not directly controlled in the simulation study. Moreover, because the mechanism generating missing values for the random variable $X_1$ differs from that used for $X_2$, $X_3$, and $X_4$ in terms of the associated probabilities, parameter recovery and imputation performance are assessed element-wise using the AAB and RMSE, as defined in \eqref{ab_and_rmse}. For imputation performance,  $\btheta$ and $\hat{\btheta}$ in AAB and RMSE are replaced with $x_{ik}$ and  $x^*_{ik}$, with $k = 1, \ldots, 3$, and are averaged over the number of observations with missing values.

\subsubsection{Simulation results: Part B - Parameter recovery}
\begin{figure}[H]
    \centering
    \includegraphics[width=\textwidth]{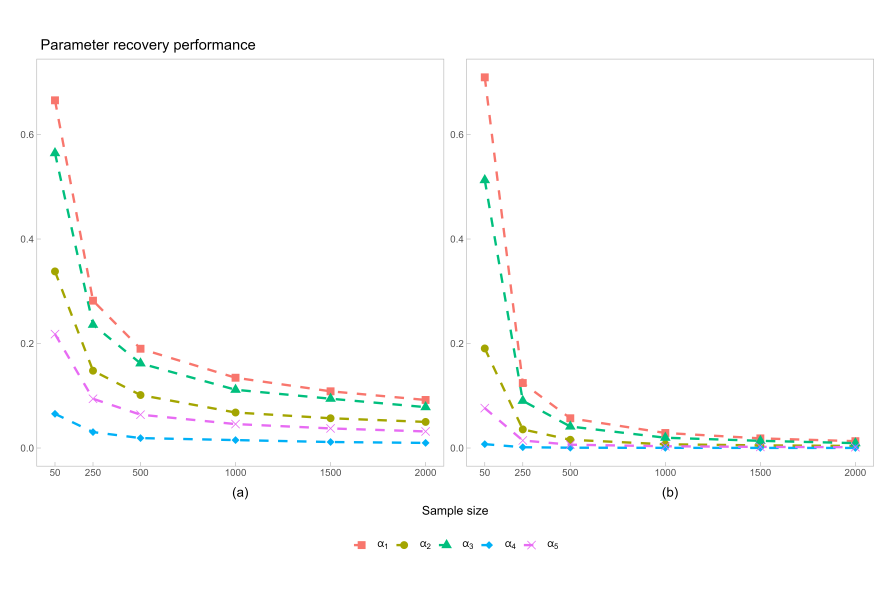}
    \caption{(a) AAB for percentage of missing values and (b) RMSE for percentage of missing values for Part B. }
    \label{alpha_mar}
\end{figure}

Figure \ref{alpha_mar} shows that both the AAB and RMSE decrease toward zero as the sample size increases, indicating that the estimates obtained using the proposed algorithm are consistent.

\subsubsection{Simulation results: Part B - imputation}

To compare the imputations between the proposed algorithm and the competitors, only observations from $X_1,X_2$, and $X_3$ are hidden as in the scheme given in \eqref{mar_scheme} to allow at least two observations for an ALR transform. The different probabilities suggest that the imputation across these three random variables might be affected differently. With that in mind, the average AAB and RMSE are calculated for $X_1,X_2$, and $X_3$ individually.
\begin{figure}[H]
    \centering
    \includegraphics[width=\textwidth]{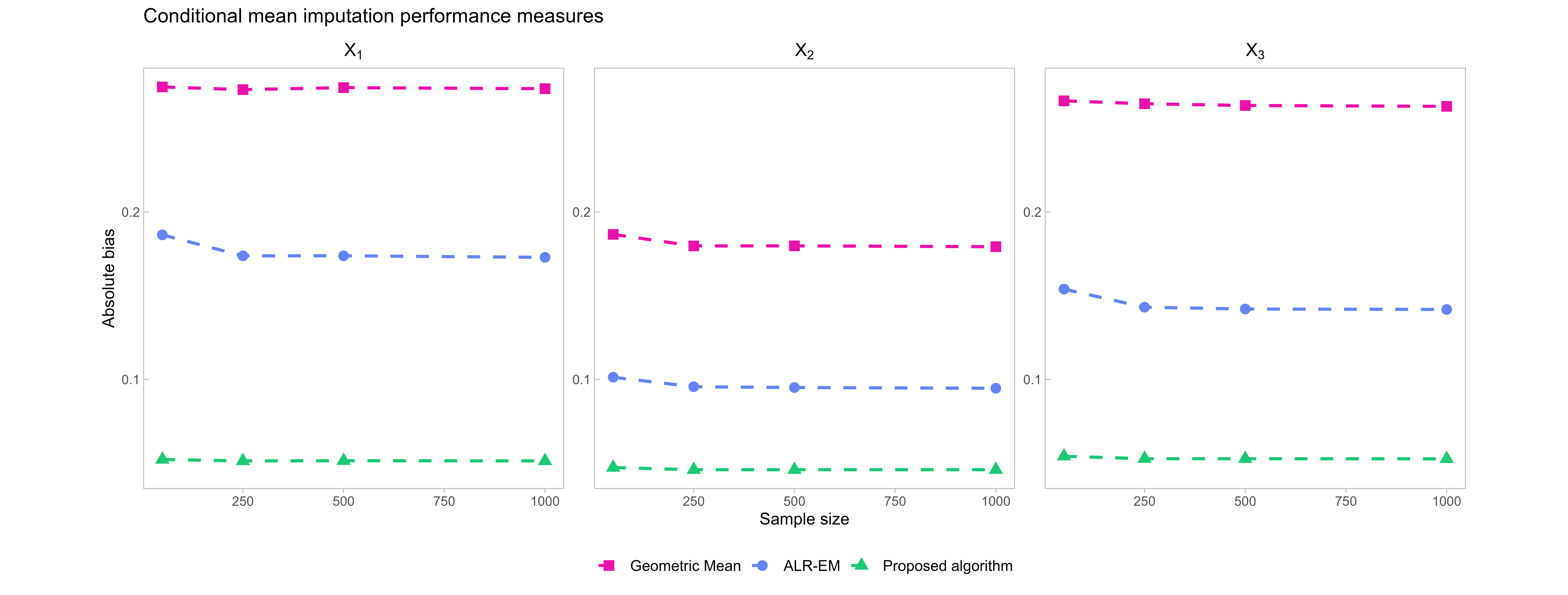}
    \caption{ Average AAB of imputed values for $X_1$, $X_2$, and $X_3$ generated by the geometric mean, ALR-EM and the proposed algorithm.}
    \label{mar_bias}
\end{figure}

\begin{figure}[H]
    \centering
    \includegraphics[width=\textwidth]{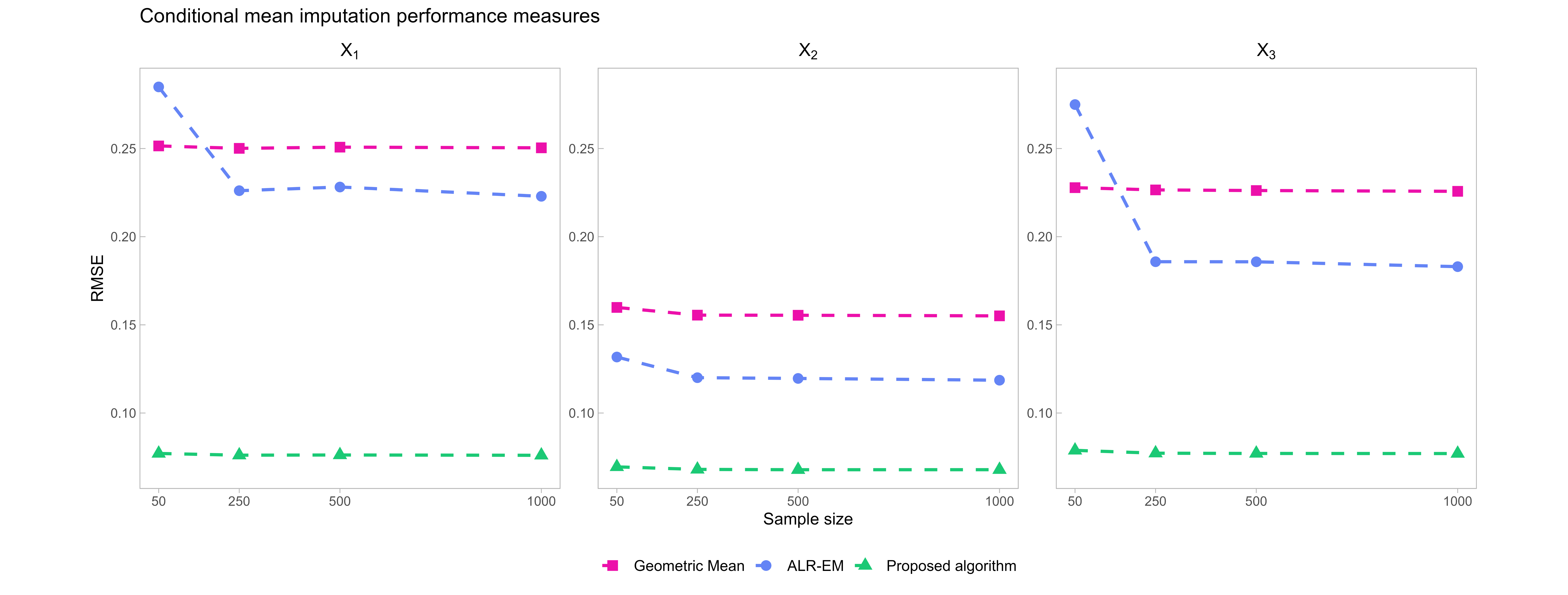}
        \caption{  Average RMSE of imputed values for $X_1$, $X_2$, and $X_3$ generated by the geometric mean, ALR-EM and the proposed algorithm.}
    \label{mar_rmse}
\end{figure}
Figures \ref{mar_bias} and \ref{mar_rmse} reveal that when data are truly missing at random, the quality of imputations is affected by the proportion of values missing from each random variable. In particular, a distribution-free technique such as geometric mean imputation ignores the dependency between the observed and missing components, leading to the large errors observed. 
\begin{figure}[H]
    \centering
    \includegraphics[width=\textwidth]{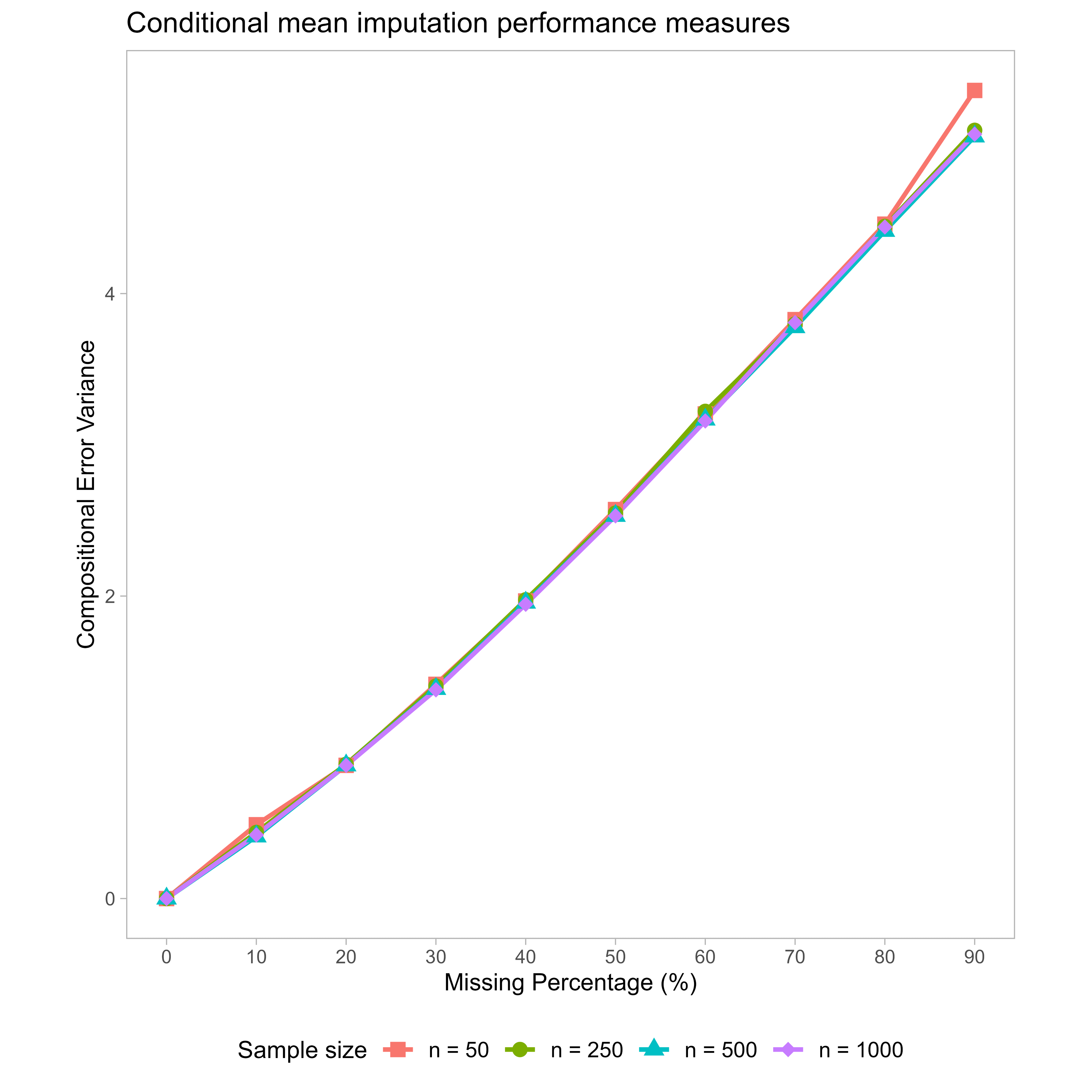}
    \caption{CEV for imputed elements.}
    \label{cev_mar}
\end{figure}
Similarly, the error metrics of the ALR-EM algorithm fluctuate across variables $X_1$ to $X_3$, although these fluctuations are less pronounced than those of the distribution-free method. In contrast, the proposed algorithm yields the lowest error metrics, with trends that remain consistent across variables. Because the dependency between the observed and missing components of each vector is explicitly accounted for, the imputation errors exhibit minimal variability.
Looking holistically, the CEVs in \figurename~\ref{cev_mar} corroborate with these findings demonstrating that the flexibility and appropriateness of the proposed algorithm's imputations persist when assessment respects the geometry of the simplex.

\subsection{Simulation results: Part C}
In this part, 1000 datasets are simulated from a Dirichlet distribution with the parameters given in \eqref{true_parm}. Censoring is imposed on the variables $X_1$ to $X_4$, based on the variable's empirical quantiles. That is
\begin{align}
\label{in_cen_region}
    q_k\left({\frac{\delta}{2}}\right)< x_k < q_k\left({1-\frac{\delta}{2}}\right)
\end{align} is observed where $q_k(\cdot)$ is a quantile corresponding to some probability $\delta$ for $k=1,\dots,p$. The threshold depends on $\delta$ and is used to control the percentage of unobserved values in the dataset. Observations are therefore hidden when they fall outside the region defined in \eqref{in_cen_region}. The simulation again focuses on parameter recovery as progressively higher levels of censoring are applied to the datasets. It also compares the imputations produced by the proposed algorithm with those obtained from the ALR-EM algorithm and geometric mean imputation.

\subsubsection{Simulation results: Part C - Parameter recovery}

\begin{figure}[H]
    \centering
    \includegraphics[width=\textwidth]{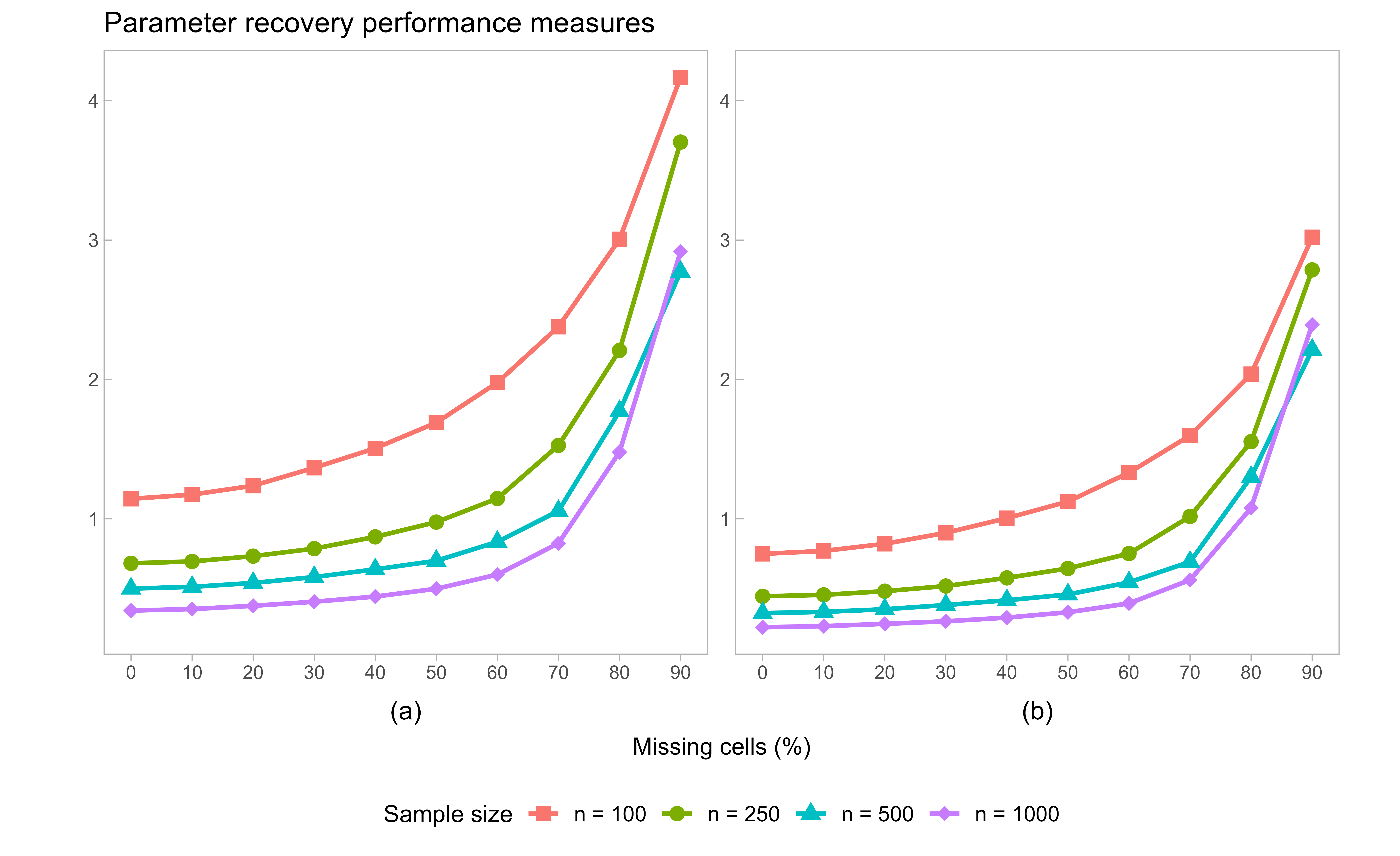}
    \caption{(a) AAB of estimated parameters for increasingly higher percentages of censored values and (b) RMSE of estimated parameters for increasingly higher percentages of censored values.}
    \label{alpha_measures_cens}
\end{figure}
As in Part A, it is also of interest to assess parameter recovery using only the fully observed rows. The results are shown in \figurename~\ref{alpha_measures_cens_complete}. At higher censoring levels, convergence was not achieved for all simulation iterations. The reported AAB and RMSE values therefore correspond only to the iterations that converged successfully. In some settings, particularly for smaller sample sizes. For instance, when the censoring percentages exceeded 80\%, the$ n=100$ and $n=250$ datasets contained insufficient fully observed rows for the Dirichlet model to be fitted. Thus, no results are shown for censored percentages above 70\%.
\begin{figure}[H]
    \centering
    \includegraphics[width=1.2\textwidth]{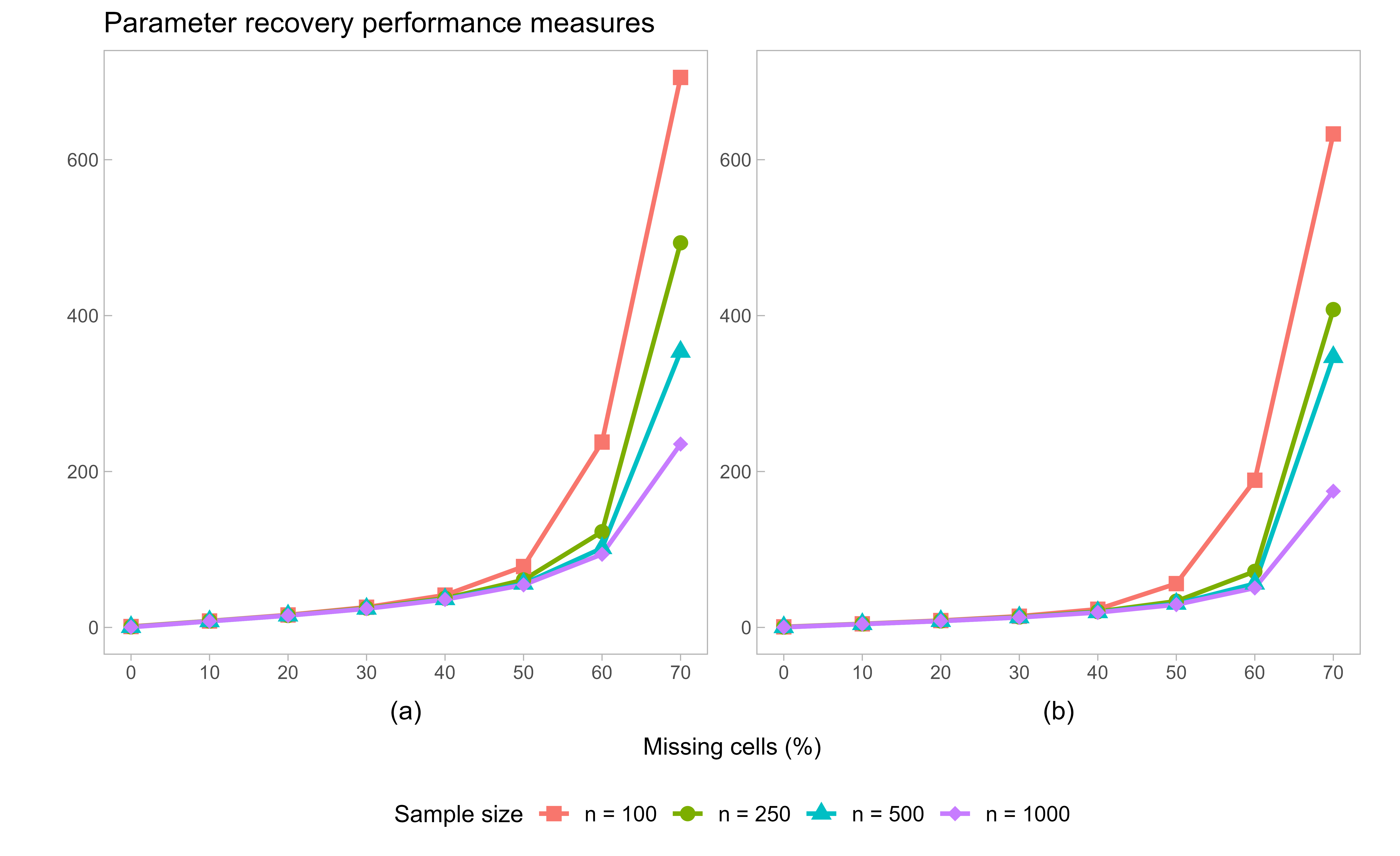}
    \caption{(a) AAB of estimated parameters for increasingly higher percentages of censored values based only on the fully observed rows and (b) RMSE of estimated parameters for increasingly higher percentages of censored values based only on the fully observed rows.}
    \label{alpha_measures_cens_complete}
\end{figure}
The seemingly exponentially increasing trends observed in \figurename~\ref{alpha_measures_cens} are consistent with those in \figurename~\ref{alpha_measures}. Regardless of the coarsening mechanism, the estimates obtained from the proposed algorithm exhibit lower bias and variability compared to those based solely on complete cases. As the percentage of unobserved observations exceeds 70\%, performance deteriorates relative to lower levels of unobserved values, which is expected. However, the increase in RMSE and AAB remains negligible compared to the corresponding complete-case scenario in \figurename~\ref{alpha_measures_cens_complete}.
\subsubsection{Simulation results: Part C - Imputation}
The ALR-EM algorithm currently only supports data subject to left-censoring. Thus, to allow for a comparison, the 1000 datasets that are simulated are censored if they fall below a threshold value, $q_k\left(\delta\right)$. In this setup, the variables $X_1$ to $X_3$ are hidden if they fall below the quantile $q_k\left(\delta\right)$ where $\delta \in \{ 0.1,0.2,0.3\}$. The CEVs for the imputations are given in Figure \ref{cev_comparison_cens}.
\begin{figure}[H]
    \centering
    \includegraphics[trim={1.6cm 0cm 1cm 1cm}, clip, width=1\linewidth]{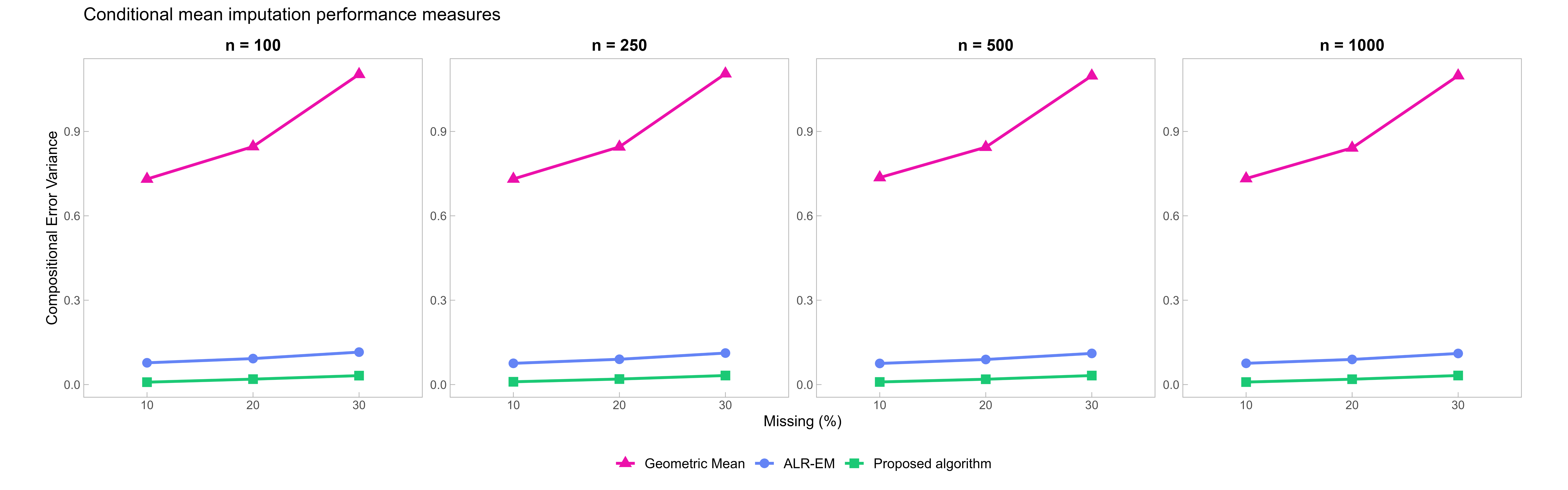}
    \caption{CEV for imputed values from geometric mean imputation, ALR-EM, and the proposed algorithm under varying percentages of censored elements.}
    \label{cev_comparison_cens}
\end{figure}
The trends of the CEVs given in Figure \ref{cev_comparison_cens} show that the proposed algorithm performs the best out of the three methods. These trends are consistent with what was seen in Part A and B under the CCAR and CAR missing mechanisms, respectively (refer to Figure \ref{cev_comparison}, Figure \ref{mar_bias} and Figure \ref{mar_rmse}).

\section{Data application: Tables and Figures}

\begin{table}[H]
    \centering
    \begin{tabular}{ccSS}
    \toprule
        Variable Name & Analyte Description	&{LLOD ($\mu_g/L)$} & {\% unobserved}\\
        \midrule
        LBXIHG	      & Mercury, inorganic &  0.21              & 83.935950     \\
        LBXBGE	      & Mercury, ethyl     &  0.064             & 99.819215     \\
        LBXBGM	      & Mercury, methyl    &  0.26              & 4.235537      \\
        LBXHGR        & Mercury, residual  &  {-}              & 99.922521      \\
        \bottomrule
    \end{tabular}
    \caption{Detection limits for the different types of mercury measured in the NHANES survey between 2021 and 2023 and percentage of unobserved (missing or censored) observations. }
    \label{llod}
\end{table}

\begin{figure}[H]
    \centering
    \fbox{\includegraphics[width=0.7\textwidth]{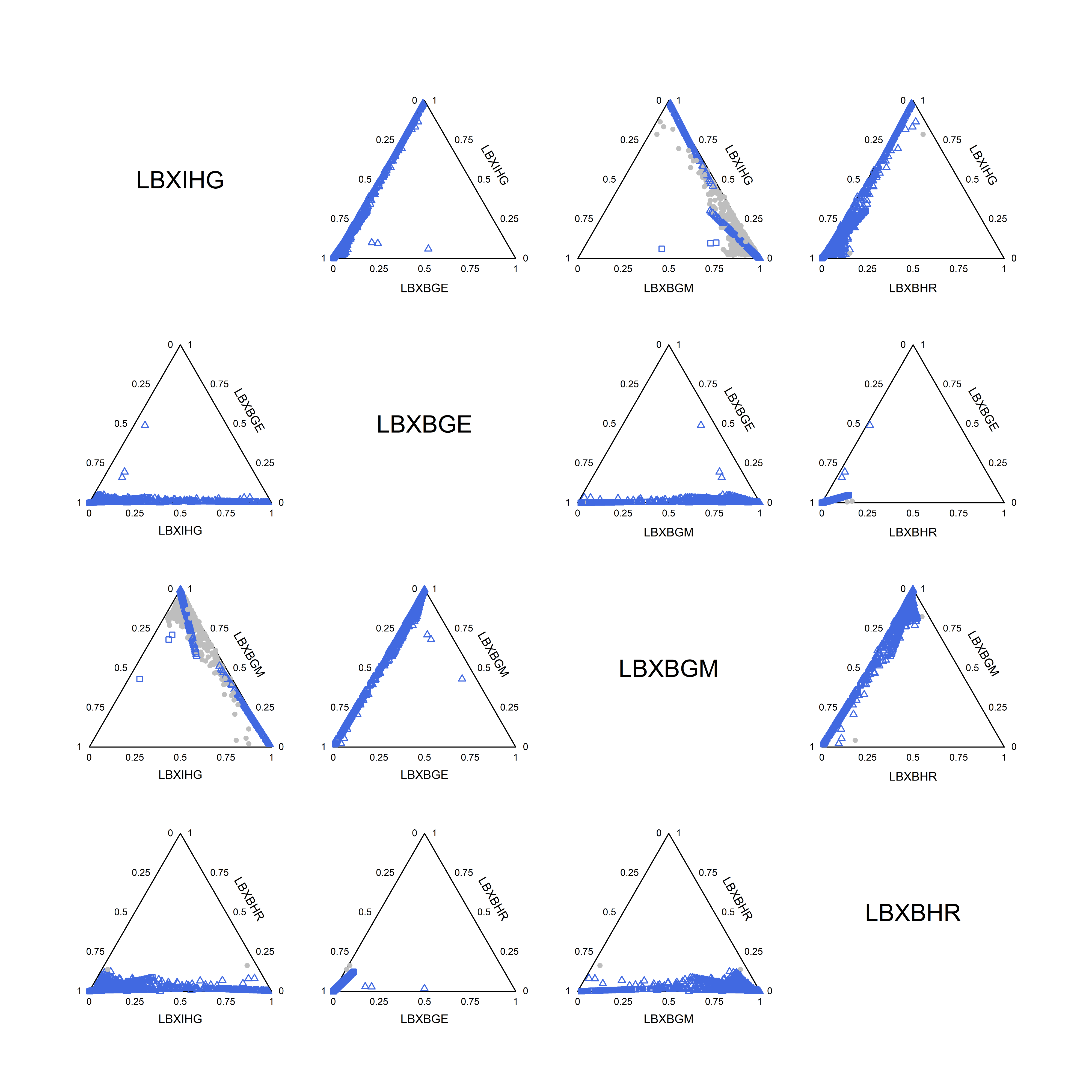}}
    \caption{Pairwise ternary plots of the types of mercury found in blood samples. The grey dots are the observed values. Blue triangles are points where one element was imputed, and blue squares are points that had two elements imputed. }
    \label{pair_plots}
\end{figure}

\begin{table}[H]
    \centering
        \begin{tabular}{lSSSS}
        \toprule
         & {LBXIHG} & {LBXBGE} & {LBXBGM} & {LBXBHR} \\
        \midrule
        $\bm{\alpha}$ & 0.8530 & 0.0805 & 6.3902 & 0.3438 \\
        Mean          & 0.1112 & 0.0105 & 0.8334 & 0.0448 \\
        \midrule
        \multicolumn{5}{c}{\text{Correlation matrix}}\\
        \midrule
        LBXIHG & 1.0000 & -0.0364 & -0.7914 & -0.0767 \\
        LBXBGE &        & 1.0000  & -0.2304 & -0.0223 \\
        LBXBGM &        &         & 1.0000  & -0.4846 \\
        LBXBHR &        &         &         & 1.0000  \\
        \bottomrule
        \end{tabular}
    \caption{Estimated parameters after fitting a Dirichlet distribution using the proposed algorithm in Section \ref{inference}. }
    \label{mle}
\end{table}

\begin{table}[H]
    \centering
        \begin{tabular}{lSSS}
        \toprule
         Distribution        & {Observed log-likelihood} & {AIC}       & {BIC} \\
        \midrule
         Multivariate normal & -573.3976                 & -1164.7952 & -1202.0152\\
         Dirichlet           & 344.4434                  &   680.8868 &  664.3445\\
        \bottomrule
        \end{tabular}
    \caption{Goodness-of-fit measures. The versions of AIC and BIC reported here are defined so that larger values indicate a better fit.}
    \label{gof}
\end{table}

\begin{table}[H]
    \centering
        \begin{tabular}{lSSS}
        \toprule
         Distribution                 & {Multivariate logistic normal} & {Dirichlet}\\
        \midrule
         Multivariate logistic normal & 0{*}                           & 14.20124 \\
         Dirichlet                    &  14.20124                      & 0{*} \\
        \bottomrule
        \end{tabular}
    \caption{Pairwise Wasserstein distances between the multivariate logistic normal and Dirichlet. \\
    * Theoretically, the Wasserstein distance between a distribution's mass and itself is zero.}
    \label{wasserstein}
\end{table}

\begin{figure}[H]
    \centering
    \includegraphics[width=\linewidth]{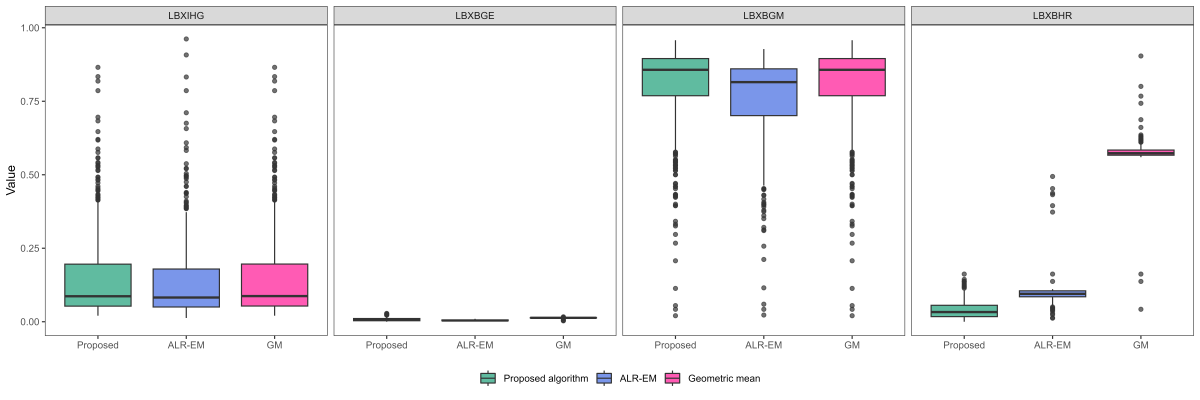}
    \caption{Marginal histograms of the imputed datasets for all four types of mercury.}
    \label{histograms}
\end{figure}

%\end{center}

%\bibitem[Mazza and Punzo(2020)]{convergence_aitken}
%Mazza, A. and Punzo, A. (2020).
%\newblock Mixtures of multivariate contaminated normal regression models.
%\newblock \emph{Statistical Papers} \textbf{61}(2), 787--822.

%\begin{description}
%\item[Title:]
%Brief description. (file type)
%\item[R-package for MYNEW routine:]
%R-package MYNEW containing code to perform the diagnostic methods
%described in the article. The package also contains all datasets used as
%examples in the article. (GNU zipped tar file)
%\item[HIV data set:]
%Data set used in the illustration of MYNEW method in
%Section~\ref{sec-verify} (.txt file).
%\end{description}

%\section{BibTeX}\label{bibtex}
%We encourage you to use BibTeX. If you have, please feel free to use the
%package natbib with any bibliography style you're comfortable with. The
%.bst file agsm has been included here for your convenience. 

\end{document}

%% file: Triangles.tex
\begin{figure}[!ht]
\centering
% --- Row 1 ---
%\resizebox{\textwidth}{!}{
\begin{subfigure}{0.48\textwidth}
\centering
\begin{tikzpicture}[ scale=2.5] 
    \coordinate (X1) at (90:1); 
    \coordinate (X2) at (210:1); 
    \coordinate (X3) at (330:1); 
    
    % Simplex 
    \draw[thick] (X1) -- (X2) -- (X3) -- cycle; 
    
    % Vertices Labels 
    \node[above] at (X1) {$x_1=1$}; 
    \node[below left] at (X2) {$x_2=1$}; 
    \node[below right] at (X3) {$x_3=1$}; 
    
    % Line x1 = 0.2 (Changed to dotted and added the label inline)
\coordinate (L1) at (barycentric cs:X1=0.2,X2=0.8,X3=0); 
\coordinate (L2) at (barycentric cs:X1=0.2,X2=0,X3=0.8); 
\draw[dashed, gray] (L1) -- (L2) 
    node[pos=0.18, sloped, fill=white, inner sep=1pt,text=black] {\scriptsize $x_1=0.2$};

    % Point at the intersection
    \coordinate (A) at (barycentric cs:X1=0.2,X2=0.4,X3=0.4); 
    \fill[blue] (A) circle (0.4pt); 
    
    % --- DOTTED CONSTRAINT LINES --- % 
    
    % x3 = 0.4 
    \coordinate (x3a1) at (barycentric cs:X1=0.6,X2=0,X3=0.4); 
    \coordinate (x3a2) at (barycentric cs:X1=0,X2=0.6,X3=0.4); 
    \draw[dashed, gray] (x3a1) -- (x3a2) 
        node[pos=0.83, sloped, fill=white, inner sep=1pt,text=black] {\scriptsize $x_3=0.4$}; 
    
    % x2 = 0.4 
    \coordinate (x2b1) at (barycentric cs:X1=0.6,X2=0.4,X3=0); 
    \coordinate (x2b2) at (barycentric cs:X1=0,X2=0.4,X3=0.6); 
    \draw[dashed, gray] (x2b1) -- (x2b2) 
        node[pos=0.19, sloped, fill=white, inner sep=1pt,text=black] {\scriptsize $x_2=0.4$}; 
    
\end{tikzpicture}

\caption{\label{fig:ex1}$\bx_1^{\top}=[0.2,0.4,0.4]$}
\end{subfigure}
\hfill
\begin{subfigure}{0.48\textwidth}
\centering
\begin{tikzpicture}[ scale=2.5] 
    \coordinate (X1) at (90:1); 
    \coordinate (X2) at (210:1); 
    \coordinate (X3) at (330:1); 
    
    % Simplex 
    \draw[thick] (X1) -- (X2) -- (X3) -- cycle; 
    
    % Vertices Labels 
    \node[above] at (X1) {$x_1=1$}; 
    \node[below left] at (X2) {$x_2=1$}; 
    \node[below right] at (X3) {$x_3=1$}; 
    
    % Line x1 = 0.2 
\coordinate (L1) at (barycentric cs:X1=0.2,X2=0.8,X3=0); 
\coordinate (L2) at (barycentric cs:X1=0.2,X2=0,X3=0.8); 
\draw[dashed, gray] (L1) -- (L2)
    node[pos=0.12, sloped, fill=white, inner sep=1pt, text=black]
    {\scriptsize $x_1=0.2$};
    % Feasible segment 
    \coordinate (A) at (barycentric cs:X1=0.2,X2=0.4,X3=0.4); 
    \coordinate (B) at (barycentric cs:X1=0.2,X2=0.3,X3=0.5); 
    \draw[line width=2pt, blue] (A) -- (B); 
    \fill[blue] (A) circle (0.4pt); 
    \fill[blue] (B) circle (0.4pt); 
    
    % --- DOTTED CONSTRAINT LINES --- % 
    
    % x2 = 0.3 (Updated from 0.35)
    \coordinate (x2a1) at (barycentric cs:X1=0.8,X2=0.2,X3=0); 
    \coordinate (x2a2) at (barycentric cs:X1=0,X2=0.2,X3=0.8); 
    \draw[dotted] (x2a1) -- (x2a2) 
        node[pos=0.15, sloped, fill=white, inner sep=1pt] {\scriptsize $x_2=0.2$}; 
    
    % x2 = 0.4 
    \coordinate (x2b1) at (barycentric cs:X1=0.6,X2=0.4,X3=0); 
    \coordinate (x2b2) at (barycentric cs:X1=0,X2=0.4,X3=0.6); 
    \draw[dotted] (x2b1) -- (x2b2) 
        node[pos=0.18, sloped, fill=white, inner sep=1pt] {\scriptsize $x_2=0.4$}; 
    
    % x3 = 0.4 
% x3 = 0.2 
\coordinate (x3a1) at (barycentric cs:X1=0.7,X2=0,X3=0.3); 
\coordinate (x3a2) at (barycentric cs:X1=0,X2=0.7,X3=0.3); 
\draw[dotted] (x3a1) -- (x3a2) 
    node[pos=0.85, sloped, fill=white, inner sep=1pt] {\scriptsize $x_3=0.3$}; 
    
    % x3 = 0.45 
    \coordinate (x3b1) at (barycentric cs:X1=0.5,X2=0,X3=0.5); 
    \coordinate (x3b2) at (barycentric cs:X1=0,X2=0.5,X3=0.5); 
    \draw[dotted] (x3b1) -- (x3b2) 
        node[pos=0.80, sloped, fill=white, inner sep=0.8pt] {\scriptsize $x_3=0.5$}; 
    
    % Feasible label 
   % \node[above right] at ($(A)!0.5!(B)$) {\scriptsize feasible}; 
\end{tikzpicture}
\caption{\label{fig:ex2}$\bx_2^{\top}=[0.2, [0.2, 0.4], [0.3, 0.5]]$}
\end{subfigure}

\vspace{0.5em}

\begin{subfigure}{0.48\textwidth}
\centering

\begin{tikzpicture}[ scale=2.5] 
    \coordinate (X1) at (90:1); 
    \coordinate (X2) at (210:1); 
    \coordinate (X3) at (330:1); 
    
    % Simplex 
    \draw[thick] (X1) -- (X2) -- (X3) -- cycle; 
    
    % Vertices Labels 
    \node[above] at (X1) {$x_1=1$}; 
    \node[below left] at (X2) {$x_2=1$}; 
    \node[below right] at (X3) {$x_3=1$}; 
    
    % Line x1 = 0.2 
    \coordinate (L1) at (barycentric cs:X1=0.2,X2=0.8,X3=0); 
    \coordinate (L2) at (barycentric cs:X1=0.2,X2=0,X3=0.8); 
\draw[dashed, gray] (L1) -- (L2)
    node[pos=0.32, sloped, fill=white, inner sep=1pt,text=black]
    {\scriptsize $x_1=0.2$};
    
    % Feasible segment 
    \coordinate (A) at (barycentric cs:X1=0.2,X2=0.4,X3=0.4); 
    \coordinate (B) at (barycentric cs:X1=0.2,X2=0.2,X3=0.6); 
    \draw[line width=2pt, blue] (A) -- (B); 
    \fill[blue] (A) circle (0.4pt); 
    \fill[blue] (B) circle (0.4pt); 
    
    % --- DOTTED CONSTRAINT LINES --- % 
    
    % x2 = 0.3 (Updated from 0.35)
    \coordinate (x2a1) at (barycentric cs:X1=0.8,X2=0.2,X3=0); 
    \coordinate (x2a2) at (barycentric cs:X1=0,X2=0.2,X3=0.8); 
    \draw[dotted] (x2a1) -- (x2a2) 
        node[pos=0.15, sloped, fill=white, inner sep=1pt] {\scriptsize $x_2=0.2$}; 
    
    % x2 = 0.4 
    \coordinate (x2b1) at (barycentric cs:X1=0.6,X2=0.4,X3=0); 
    \coordinate (x2b2) at (barycentric cs:X1=0,X2=0.4,X3=0.6); 
    \draw[dotted] (x2b1) -- (x2b2) 
        node[pos=0.18, sloped, fill=white, inner sep=1pt] {\scriptsize $x_2=0.4$};

    % Feasible label 
   % \node[above right] at ($(A)!0.5!(B)$) {\scriptsize feasible}; 
\end{tikzpicture}

\caption{\label{fig:ex3}$\bx_3^{\top}=[0.2, [0.2,0.4], \left[0,1\right]]$}
\end{subfigure}
\hfill
\begin{subfigure}{0.48\textwidth}
\centering

\begin{tikzpicture}[ scale=2.5]

\coordinate (X1) at (90:1);
\coordinate (X2) at (210:1);
\coordinate (X3) at (330:1);

% Simplex
\draw[thick] (X1) -- (X2) -- (X3) -- cycle;

% Labels
\node[above] at (X1) {$x_1=1$};
\node[below left] at (X2) {$x_2=1$};
\node[below right] at (X3) {$x_3=1$};

% Endpoints of x1 = 0.2 slice
\coordinate (A) at (barycentric cs:X1=0.2,X2=0.8,X3=0);
\coordinate (B) at (barycentric cs:X1=0.2,X2=0,X3=0.8);

% Highlight full feasible segment
\draw[line width=2pt, blue] (A) -- (B);
\fill[blue] (A) circle (0.4pt);
\fill[blue] (B) circle (0.4pt);

\node[above] at ($(A)!0.5!(B)$)
{\small $x_1=0.2$};

\end{tikzpicture}
\caption{\label{fig:ex4}$\bx_4^{\top}=[0.2,\left[0,1\right],\left[0,1\right]]$}
\end{subfigure}
%}
%\hfill
\caption{\label{fig:example}
Graphical representation on $\mathbb{S}_3$ of four compositions.
%with varying degrees of missing and censored components. 
The thick blue segment (or region) represents the set of feasible values consistent with the observed information and the compositional constraint.}
\end{figure}